\documentclass[runningheads]{llncs}
\usepackage[utf8]{inputenc}
\usepackage{graphicx}
\usepackage{amssymb}
\usepackage{amsmath, mathtools}
\usepackage{nicefrac}
\usepackage{bbm}
\usepackage{array}
\usepackage[noend]{algpseudocode}
\usepackage{color}
\usepackage{tikz}
\usepackage{enumitem}
\newtheorem{observation}[theorem]{Observation}

\usepackage{xcolor}

\begin{document}

\title{Total Completion Time Scheduling\\ Under Scenarios}

\authorrunning{Total completion time scheduling under scenarios} 




%

\newcommand{\Sk}[1]{\marginpar{\color{cyan}\footnotesize MS: #1}}

\newcommand{\Ls}[1]{\marginpar{\color{blue}\footnotesize LS: #1}}

\newcommand{\MvE}[1]{\marginpar{\color{magenta}\footnotesize MvE: #1}}

\newcommand{\EE}[1]{\marginpar{\color{red}\footnotesize Ekin: #1}}
\newcommand{\K}{K}

\author{Thomas Bosman\inst{1} \and
        Martijn van Ee\inst{2} \and Ekin Ergen\inst{3} \and Csan\'ad Imreh\inst{4} \and
				Alberto Marchetti-Spaccamela\inst{5,6} \and
				Martin Skutella\inst{3} \and
				Leen Stougie\inst{1,6,7}
}

\authorrunning{Bosman et al.} 

\institute{Vrije Universiteit Amsterdam, The Netherlands\\
\email{tbosman@gmail.com}
\and
Netherlands Defence Academy, Den Helder, The Netherlands\\
\email{M.v.Ee.01@mindef.nl}\\
\and
Technische Universit\"at Berlin, Germany \\
\email{ergen@math.tu-berlin.de, martin.skutella@tu-berlin.de}
 \and
Our co-author Csan\'ad Imreh tragically passed away on January 5th, 2017.
\and
Universit\'a di Roma ``La Sapienza'', Italy\\
\email{alberto@diag.uniroma.it}
\and
CWI, Amsterdam, The Netherlands, \and 
Erable, INRIA, France \\
\email{leen.stougie@cwi.nl}
}


\maketitle

\begin{abstract}
Scheduling jobs with given processing times on identical parallel machines so as to minimize their total completion time is one of the most basic scheduling problems. We study interesting generalizations of this classical problem involving scenarios. In our model, a scenario is defined as a subset of a predefined and fully specified set of jobs. The aim is to find an assignment of the whole set of jobs to identical parallel machines such that the schedule, obtained for the given scenarios by simply skipping the jobs not in the scenario, optimizes a function of the total completion times over all scenarios.

While the underlying scheduling problem without scenarios can be solved efficiently by a simple greedy procedure (SPT rule), scenarios, in general, make the problem NP-hard. We paint an almost complete picture of the evolving complexity landscape, drawing the line between easy and hard. One of our main algorithmic contributions relies on a deep structural result on the maximum imbalance of an optimal schedule, based on a subtle connection to Hilbert bases of a related convex cone.
\end{abstract}



\section{Introduction} 

For a set~$J$ of $n$ jobs with given processing times~$p_j$, $j\in J$, one of the oldest results in scheduling theory states that scheduling the jobs in order of non-decreasing processing times on identical parallel machines {in a round robin procedure} minimizes the total completion time, that is, the sum of completion times of all jobs~\cite{ConwayMaxwellMiller1967}. 

We study an interesting generalization of this classical scheduling problem where the input in addition specifies a set of~$\K$ scenarios~$\mathcal{S}=\{S_1,\dots,S_\K\}$, with~$S_k\subseteq J$, $k=1,\dots,\K$. The task is then to find, for the entire set of jobs~$J$, a parallel machine schedule, which is an assignment of all jobs to machines and on each machine an order of the jobs assigned to it. This naturally induces a schedule for each scenario by simply skipping the jobs not in the scenario. In particular, jobs not contained in a particular scenario do \emph{not} contribute to the total completion time of that scenario 
and, in particular, do \emph{not} delay later jobs assigned to the same machine.

We aim to find a schedule for the entire set of jobs that optimizes a function of the total completion times of the jobs over all scenarios.  More specifically, we focus on two functions on the scheduling objectives: in the 
MinMax version, we minimize the maximum total completion time over all scenarios, and in the MinAvg version we minimize the average of the total completion times over all scenarios. In the remainder of the paper we refer to the MinMax version as MinMaxSTC (MinMax Scenario scheduling with Total Completion time objective) and to the MinAvg version as MinAvgSTC. 

\medskip
\noindent
\textbf{Optimization under scenarios.}
Scenarios are commonly used in optimization to model uncertainty in the input or different situations that need to be taken into account. A variety of approaches has been proposed that appear in the literature under different names. In fact, scenarios have been introduced to model discrete distributions over parameter values in Stochastic Programming~\cite{birge2011introduction}, or as samples in Sampling Average Approximation algorithms for stochastic problems with continuous distributions over parameter values~\cite{kleywegt2002sample}. 
In Robust Optimization~\cite{Ben-TalGN09}, scenarios describe different situations that should be taken into account and are often specified as ranges for parameter values that may occur. Moreover, in data-driven optimization, scenarios are often obtained as observations. The problems we consider also fit in the general framework of A Priori Optimization~\cite{BertsimasJailletOdoni1990}: the schedule for the entire set of jobs can be seen as an a priori solution which is updated in a simple way to a solution for each scenario.  In the scheduling literature, different approaches to modeling scenarios have been introduced, for which we refer to a very recent overview by Shabtay and Gilenson~\cite{shabtay2022}.
Another related and popular framework is that of Min–Max Regret, aiming at obtaining a solution minimizing the maximum deviation, over all possible scenarios, between the value of the
solution and the optimal value of the corresponding scenario \cite{KouvelisYu97}.

Not surprisingly, for many problems,  multiple scenario versions are fundamentally harder than their single scenario counterparts. Examples are the shortest path problem with a scenario specified by the destination~\cite{immorlica2004costs,KouvelisYu97}, and the metric minimum spanning tree problem with a scenario defining the subset of vertices to be in the tree~\cite{BertsimasJailletOdoni1990}. Scenario versions of NP-hard combinatorial optimization problems were also considered in the literature such as, for example, set cover~\cite{adamczyk2017optimum} and the traveling salesperson problem~\cite{Ee2018priori}.

\medskip
\noindent

\medskip
\noindent
\textbf{Related work.}
As we have already discussed above, a variety of approaches to optimization under scenarios appear in the literature under different names. Here we mention work in the field of scheduling that is more closely related to the model considered in this paper. We refer to  the survey~\cite{shabtay2022} and references therein for an overview of scheduling under scenarios. 

Closest to the problems considered in this paper is the work of Feuerstein et al.~\cite{feuerstein2016minimizing} who also consider scenarios given by subsets of jobs and develop approximation algorithms as well as non-approximability results for minimizing the makespan on identical parallel machines, both for the MinMax and the MinAvg version. In fact, the hardness results for our problem given in Proposition~\ref{prop:ApproxLowerBounds} below follow directly from their work.

In \emph{multi-scenario models}, a discrete set of scenarios is given, and certain parameters (e.g., processing times) of jobs can have different values in different scenarios. 
Several papers follow this model, mainly focusing on single machine scheduling problems. Various functions of scheduling objectives over the scenarios are considered, that have the MinMax and the MinAvg versions as special cases. Yang and Yu~\cite{yang2002robust}, for example, study a multi-scenario model and show that the MinMax version of minimizing total completion time is NP-hard even on a single machine and with only~$2$ scenarios, whereas in our model $2$-scenario versions are generally easy. 
(Notice that our model is different from simply assigning a processing time of $0$ to a job in a scenario if the job is not present in that scenario.)
Aloulou and Della Croce~\cite{AloulouC08} present algorithmic and computational complexity results for several single machine scheduling problems. Mastrolilli, Mutsanas, and Svensson~\cite{mastrolilli2013single} consider the MinMax version of minimizing the weighted total completion time on a single machine and prove interesting approximability and inapproximability results. Kasperski and Zieli\'nski~\cite{kasperskiZielinski} 
consider a more general single machine scheduling problem in which precedence constraints between jobs are present and propose a general framework for solving such problems with various objectives. 

Kasperski, Kurpisz, and Zieliński~\cite{KasperskiZielinski12} study 
multi-scenario scheduling with parallel machines and the makespan objective function, where the processing time of each job depends on the scenario; they give approximability results for an unbounded number of machines and a constant number of scenarios.  Albers and Janke~\cite{Albers21} as well as Bougeret, Jansen, Poss, and Rohwedder~\cite{Bougeret21} study a budgetary model with uncertain job processing times; in this model, each job has a regular processing time while in each scenario up to~$\Gamma$ jobs fail and require additional processing time. The considered objective function  is to minimize the makespan: \cite{Bougeret21} proposes approximate algorithms for identical and unrelated parallel machines while \cite{Albers21} analyses online algorithms in this setting.




\section{Preliminaries}

We start by defining the problem formally, denoting the set of integers~$\{1,\ldots,\ell\}$ simply by~$[\ell]$. We are given a set of jobs $J=[n]$, machines $M=[m]$, non-negative job processing times~$p_j$, $j\in [n]$, and scenarios $\mathcal{S}=\{S_1,\ldots,S_\K\}$, where $S_k\subseteq J$, $k\in[K]$. We assume that the jobs are ordered by \emph{non-increasing} processing times~$p_1\geq p_2\geq\dots\geq p_n$.\footnote{In view of the SPT rule, this ordering might seem counterintuitive. But it turns out to be convenient as we argue below in Observation~\ref{obs:reverse}.} The task is to find a machine assignment, that is, a map $\varphi\colon [n]\to [m]$, or equivalently, a partitioning of the jobs~$J_1,\ldots,J_m$ with the understanding that jobs in~$J_i$ shall be optimally scheduled on machine~$i\in[M]$, that is, according to the \emph{Shortest Processing Time first} (SPT) rule (i.e., in reverse order of their indices). Thus, the completion time of a particular job~$j'\in J_i$ in scenario~$k$ is the sum over all processing times of jobs~$j\in J_i\cap S_k$ with~$j\geq j'$, and the contribution of jobs in~$J_i$ to the total completion in scenario~$k\in[\K]$ is thus:
\begin{align}
\sum_{j'\in J_i\cap S_k}~
\sum_{j\in J_i\cap S_k: j'\leq j}p_j
=
\sum_{j\in J_i\cap S_k} 
p_j\cdot|\{j'\in J_i\cap S_k: j'\leq j\}|
\label{eq:exchange}
\end{align}

\begin{observation}[\cite{EastmanEven}]
\label{obs:reverse}
The above total unweighted completion time problem is equivalent to the total weighted completion time with jobs of unit length of weight \mbox{$w_j\coloneqq p_j$} that are processed in reverse order, i.e., in order of \emph{non-increasing} weights. 
\end{observation}

Indeed, this equivalence between total completion time for unweighted jobs scheduled in SPT order and total \emph{weighted} completion time for unit length jobs scheduled in order of non-decreasing weight was first observed by Eastman, Even, and Issacs~\cite{EastmanEven}. The idea behind the equivalence~\eqref{eq:exchange} is best seen from a so-called 2-dimensional Gantt-chart; see Figure \ref{fig:2DG} 
\begin{figure}[t]
\centering
\includegraphics[scale=0.80]{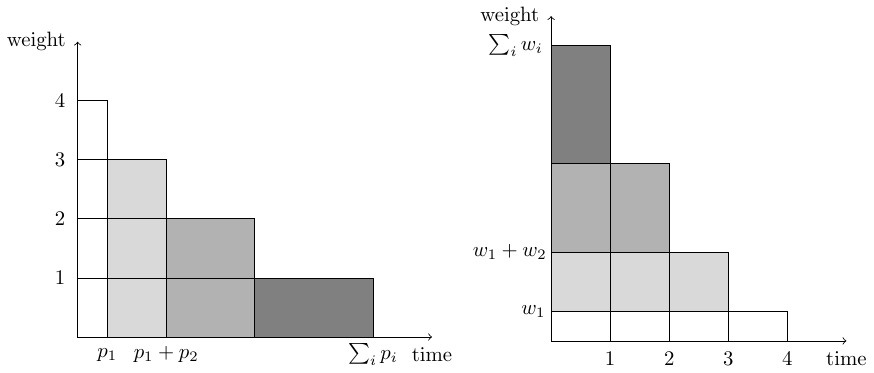}
\caption{Left: original schedule. Right: equivalent ``weight-schedule''. In both cases the objective value is equal to the total area of the rectangles.}
\label{fig:2DG}
\end{figure}
and the work of Goemans and Williamson~\cite{GoemansW00}, Megow and Verschae~\cite{MegowVerschae2018}, or Cho, Shmoys, and Henderson~\cite{ChoSH23}. 

In the introduction of this paper we prefer to present our results in terms of the more commonly known unweighted version
, but for the remainder of this paper it is somewhat easier to argue about the weighted unit-processing time version; hence the ordering of jobs introduced above.
The objective of (MinMaxSTC) is then to minimize
\[
\max_{k\in[K]} 
\sum_{i=1}^m 
\sum_{j\in J_i\cap S_k}
w_j\cdot|\{j'\in J_i\cap S_k: j'\leq j\}|,
\] 
whereas MinAvgSTC aims to minimize
\[
\frac{1}{K} \sum_{k=1}^\K 
\sum_{i=1}^m 
\sum_{j\in J_i\cap S_k}
w_j\cdot|\{j'\in J_i\cap S_k: j'\leq j\}|.
\] 

In the sequel in the MinAvgSTC we neglect the constant $1/K$-term and minimize the sum over all scenarios. 

\section{Our Contribution}

We give a nearly complete overview of the complexity landscape of total completion time scheduling under scenarios. 
Tables~\ref{summary-result-2} and~\ref{summary-result-1} summarize our observations for MinMaxSTC and MinAvgSTC, respectively. The rows of both tables correspond to different assumptions on the number of scenarios~$\K$, whereas the columns specify assumptions on the given number of machines~$m$. 

\newcolumntype{A}{>{\centering}p{0.3\textwidth}}
\renewcommand{\arraystretch}{2}
\begin{table}[t] 
\begin{center}
\begin{tabular}{lAcc}
\hline
& $m=2$ & $m\in O(1)$ & $m$ part of input\\ \hline
$\K=2$ & poly & poly & poly\\
\rule{0pt}{22pt}$3\leq\K\in O(1)$ & 
\begin{minipage}{20ex}
\begin{center}
weakly NP-hard,\\
pseudo-poly,\\
FPTAS
\end{center}
\end{minipage}
& 
\begin{minipage}{20ex}
\begin{center}
weakly NP-hard,\\
pseudo-poly,\\
FPTAS
\end{center}
\end{minipage}
& 
\begin{minipage}{20ex}
\begin{center}
weakly NP-hard,\\ 
poly if $w_j\in O(1)$
\end{center}
\end{minipage}
\\
\rule{0pt}{27pt}$\K$ {part of input} & 
\begin{minipage}{25ex}
\begin{center}
strongly NP-hard~\cite{feuerstein2016minimizing},\\
no~$(2-\varepsilon)$-approx~\cite{feuerstein2016minimizing},\\
$2$-approx
\end{center}
\end{minipage}
&\begin{minipage}{20ex}
\begin{center}
    
 strongly NP-hard \\
\cite{feuerstein2016minimizing}
\end{center}
\end{minipage}
&\begin{minipage}{20ex}
\begin{center}
    
 strongly NP-hard \\
\cite{feuerstein2016minimizing}
\end{center}
\end{minipage}\\[1ex] \hline
\end{tabular}
\end{center}
\caption{The complexity landscape of MinMaxSTC\label{summary-result-2} on $m$ machines with $K$ scenarios}
\end{table}

First of all, it is not difficult to observe that both MinMaxSTC and MinAvgSTC are strongly NP-hard if~$\K$ can be arbitrarily large; see last row of Tables~\ref{summary-result-2} and~\ref{summary-result-1}. This even holds for the special case of unit length jobs and only two jobs per scenario. Moreover, for the case of  MinMaxSTC on two machines, we get a tight non-approximability result and corresponding approximation algorithm, while for MinAvgSTC we can prove that the problem is APX-hard, i.e., there is no PTAS, unless P=NP; see Section~\ref{sec:anysce}. For only~$\K=2$ scenarios, however, both problems can be solved to optimality in polynomial time; see first row of Tables~\ref{summary-result-2} and~\ref{summary-result-1}. Even better, in Section~\ref{sec:anysce} we present a simple algorithm that constructs an `ideal' schedule for the entire set of jobs simultaneously minimizing the total completion time in both scenarios. These results develop a clear complexity gap between the case of two and arbitrarily many scenarios.

A finer distinction between easy and hard can thus be achieved by considering the case of constantly many scenarios~$\K\geq3$; see middle row in Tables~\ref{summary-result-2} and~\ref{summary-result-1}. These results constitute the main contribution of this paper.

Our results on MinMaxSTC for constantly many scenarios~$\K\geq3$ are presented in Section~\ref{sec:minmax}. Here it turns out that MinMaxSTC is weakly NP-hard already for~$\K=3$ scenarios and~$m=2$ machines, but can be solved in pseudo-polynomial time for any constant number of scenarios and machines via dynamic programming. Moreover, the dynamic program together with standard rounding techniques immediately implies the existence of an FPTAS for this case. If the number of machines~$m$ is part of the input, however, our previous dynamic programming approach fails. But then again, we present a more sophisticated dynamic program that solves the problem efficiently if all job processing times are bounded by a constant. 

MinAvgSTC with constantly many scenarios~$\K\geq3$ is studied in Section~\ref{sec:minsum}. Somewhat surprisingly, and in contrast to MinMaxSTC, it turns out that MinAvgSTC remains easy as long as the number of machines~$m$ is bounded by a constant. This observation is again based on a dynamic programming algorithm. Moreover, we conjecture that the problem even remains easy if~$m$ is part of the input. More precisely, we conjecture that there always exists an optimal solution such that the imbalance between machine loads always remains bounded by~$g(\K)$ for some (exponential) function~$g$ that only depends on the number of scenarios~$\K$, but not on~$m$ or~$n$. Using a subtle connection to the cardinality of Hilbert bases for convex cones, we prove that our conjecture is true for unit job processing times. In this case we obtain an efficient algorithm with running time~$m^{h(\K)}\cdot poly(n)$ for some function~$h$. 

\renewcommand{\arraystretch}{2}
\begin{table}[t] 
\begin{center}
\begin{tabular}{lccc}
\hline
& $m=2$ & $m\in O(1)$ & $m$ part of input\\ \hline
$\K=2$ & poly & poly & poly \\
$3\leq\K\in O(1)$ & poly & poly & 
\begin{minipage}{20ex}
\begin{center}
poly if $w_j\in O(1)$
\end{center}
\end{minipage}
\\
\rule{0pt}{27pt}$\K$ {part of input} & 
\begin{minipage}{20ex}
\begin{center}
strongly NP-hard\\ \cite{feuerstein2016minimizing},\\
no PTAS~\cite{haastad2001some,khot2007optimal},\\
\nicefrac54-approx~\cite{SchulzS02,Skut-Diss98}
\end{center}
\end{minipage}
& 
\begin{minipage}{25ex}
\begin{center}
strongly NP-hard~\cite{feuerstein2016minimizing},\\
no PTAS~\cite{haastad2001some,khot2007optimal},\\
(\nicefrac32-\nicefrac1{2m})-approx\\ \cite{SchulzS02,Skut-Diss98}
\end{center}
\end{minipage}
& 
\begin{minipage}{25ex}
\begin{center}
strongly NP-hard~\cite{feuerstein2016minimizing},\\
no PTAS~\cite{haastad2001some,khot2007optimal},\\
(\nicefrac32-\nicefrac1{2m})-approx\\ \cite{SchulzS02,Skut-Diss98}
\end{center}
\end{minipage}
\\[2.5ex] \hline
\end{tabular}
\end{center}
\caption{The complexity landscape of MinAvgSTC\label{summary-result-1} on $m$ machines with $K$ scenarios}
\end{table}

Several of our results, in particular for MinMaxSTC, can be generalized to the Min–Max Regret framework; see Appendix~\ref{subsec:regret} 
for details. Moreover, due to space restrictions, most proofs are deferred to the appendix. 

\section{Scheduling Under Arbitrary $K$ is Hard, but $\K=2$ is Easy}
\label{sec:anysce}

\subsection{NP-hardness for unbounded number of scenarios}

For an unbounded number of scenarios, there is a straightforward proof that both MinMaxSTC and MinAvgSTC are NP-hard on~$m\geq3$ machines which relies on a simple reduction of the graph coloring problem: Given a graph, we interpret its nodes as unweighted unit length jobs and every edge as one scenario consisting of the two jobs that correspond to the end nodes of the edge. Obviously, the graph has an $m$-coloring without monochromatic edges if and only if there is a schedule such that the total completion time of each scenario is~$2$ (i.e., both jobs complete at time~$1$ on a machine of their own).

Since it is easy to decide whether the nodes of a given graph can be colored with two colors, the above reduction does not imply NP-hardness for~$m=2$ machines. For this case, however, the inapproximability results for the multi scenario makespan problem in~\cite{feuerstein2016minimizing} can be adapted to similar inapproximability results for our problems. Also these results already hold for unweighted unit length jobs. The proof is deferred to Appendix~\ref{subsec:approx-proofs}.

\begin{proposition}\label{prop:ApproxLowerBounds}
For two machines and all jobs having unit lengths and  weights, it is NP-hard to approximate MinMaxSTC within a factor $2-\varepsilon$ and MinAvgSTC within ratio 1.011. The latter even holds if all scenarios contain only two jobs.
\end{proposition}

We notice that, for~$m=2$ machines, any algorithm for MinMaxSTC that assigns all jobs to the same machine (in SPT order) gives a $2$-approximation. The approximability of MinMaxSTC for more than two machines is left as an interesting open question. The following approximation result for MinAvgSTC follows from a corresponding result for classical machine scheduling without scenarios; see Appendix~\ref{subsec:approx-proofs} 
for a short proof.

\begin{proposition}\label{prop:ApproxUpperBound}
For MinAvgSTC there is a $(\nicefrac32-\nicefrac{1}{2m})$-approximation algorithm for arbitrarily many machines, even if~$m$ is part of the input.
\end{proposition}

\subsection{Computing an ideal schedule for two scenarios}

For~$\K=2$ scenarios, both MinMaxSTC and MinAvgSTC are polynomial time solvable on any number of machines. Actually, we prove an even stronger result: one can find, in polynomial time, a schedule that has optimal objective function value in each of the two scenarios simultaneously. 

\begin{theorem}
\label{master}
For the MinMaxSTC and the MinAvgSTC problem with two scenarios, after sorting the jobs in order of non-increasing weight, one can find in 
time linear in $n$
a schedule that is simultaneously optimal for both scenarios.
\end{theorem}
\begin{proof}
We show how to assign the jobs in \emph{non-increasing} order of their weights in order to be optimal for both scenarios. As mentioned above, we want to assign the next job to a machine that, in each scenario, belongs to the least loaded machines in terms of the number of jobs already assigned to it. For this purpose we define two $m$-dimensional vectors~$s_1$ and~$s_2$ for scenarios~$1$ and~$2$, respectively, containing the relative loads on the machines. The relative load of a machine in a scenario is~$0$ if this machine belongs to the least loaded machines in this scenario; it is~$1$ if it has one job more than a least loaded machine, etc. In our assignment process, jobs will always be assigned to machines with relative load~$0$ in each of the scenarios they belong to. This ensures that we will end up with a schedule that is optimal for both scenarios simultaneously.

Initially, both vectors~$s_1$ and~$s_2$ are zero vectors, since no jobs have been assigned yet. When assigning job~$j$, let~$\mu_k$ be the lowest entry (lowest numbered machine) equal to zero in vector~$s_k$. Moreover, let~$\nu_k$ be the highest entry equal to zero in~$s_k$. So initially,~$\mu_1=\mu_2=1$ and~$\nu_1=\nu_2=m$. We apply the following assignment procedure, where we use $\mathbbm{1}$ to denote the all-$1$ vector.
For $j=1$ to $n$, we apply the following case distinction:
\begin{itemize}
 \item If $j\in S_1\setminus S_2$ $\rightarrow$ assign~$j$ to machine $\mu_1$ and increase $\mu_1$ by~$1$.
 \item If $j\in S_2\setminus S_1$ $\rightarrow$ assign~$j$ to machine~$\mu_2$ and increase~$\mu_2$ by~$1$.
 \item If $j\in S_1\cap S_2$ $\rightarrow$ assign~$j$ to machine $\nu_1=\nu_2$ and decrease both $\nu_1$ and~$\nu_2$ by~$1$.
 \item If $s_1=\mathbbm{1}$ $\rightarrow$ reset $s_1$ to become the all-$0$ vector. Reindex the machines such that~$s_2$ becomes of the form 
$(1,\ldots,1,0,\ldots,0)$. Reset $\mu_1=1$, $\nu_1=m$, $\mu_2=\mu_2+m-\nu_2$, and $\nu_2=m$. Do analogously if $s_2=\mathbbm{1}$.
\end{itemize}
We prove that for each job there is always a machine with
relative load~$0$ in each scenario in which the job appears, thus implying the theorem. This is obviously true if job~$j$ appears in
only one scenario, since after assigning~$j-1$ jobs,~$s_k\neq
\mathbbm{1}$, $k=1,2$. Hence there is always a machine with relative
load~$0$. For job~$j$ appearing in both scenarios, we have to show
that we maintain~$\nu_1=\nu_2$, in which case the same machine has
relative load~$0$ to accommodate job~$j$. The only way it can happen
that~$\nu_1\neq \nu_2$ is if ever machine $\nu_1$ was used for a job
$j'$ that only appeared in scenario~$1$. But that can only have happened
if~$\mu_1 = \nu_1$, in which case~$s_1$ becomes an all-$1$ vector, and by resetting it to~$0$ and the renumbering of the machines,
the relation~$\nu_1=\nu_2$ had been restored.\qed
\end{proof}



\section{The MinMax Version}
\label{sec:minmax}

\subsection{Constant number of machines}

In view of Theorem~\ref{master}, the next theorem establishes a complexity gap for MinMaxSTC when going from two to three scenarios.

\begin{theorem}\label{partitionhardness}
On any fixed number $m\geq 2$ of machines and with~$K=3$ scenarios, MinMaxSTC is (weakly) NP-hard.
\end{theorem}

The proof can be found in Appendix~\ref{appendix:MinMaxSTC}. 
It reduces a variant of Partition, in which one is to partition a set of numbers into three sets of equal sum instead of two. This reduction establishes that MinMaxSTC is weakly NP-hard. 
For the case $m\in O(1)$, the weak NP-hardness cannot be strengthened to strong NP-hardness (unless P=NP) since on a fixed number of machines and scenarios an optimal solution can be found in pseudopolynomial time by dynamic programming.
Moreover, via standard rounding techniques, one can obtain an FPTAS for MinMaxSTC. 
Details are given in Appendix~\ref{appendix:MinMaxSTC}.

\begin{theorem}\label{theorem:fptas}
There exists a pseudopolynomial algorithm as well as a fully polynomial time approximation scheme for MinMaxSTC on a constant number of machines with a constant number of scenarios.   
\end{theorem}

The dynamic program runs in time $O(m(m^2n^3W)^{mK})$, where $W$ denotes the largest (integer) job weight. The rounding for FPTAS redefines this weight to be $1+\nicefrac{mn^2}{\varepsilon}$, i.e., the runtime of the FPTAS is in $O\bigl(m\bigl(m^2n^3(\nicefrac{mn^2}{\varepsilon})\bigr)^{mK}\bigr)$, indeed polynomial in $n$ and $\nicefrac{1}{\varepsilon}$ assuming that~$m$ and~$K$ are constant. 

An immediate consequence of the reduction that yields Theorem~\ref{partitionhardness} is the NP-hardness of the robust version of scheduling with regret, as mentioned in~\cite{shabtay2022}, even when it is restricted to the parallel machine case. Similarly, the FPTAS given in Theorem~\ref{theorem:fptas} can easily be adapted to this model. A more elaborate discussion on the relations between our model and the scheduling with regret model can be found in Appendix~\ref{subsec:regret}.

\subsection{Any number of machines}

If job weights are bounded by a constant, MinMaxSTC (and also MinAvgSTC) can be solved efficiently on any number of machines by dynamic programming. For simplicity, we only discuss the case of unit job weights here, but our approach can be easily generalized to the case of weights bounded by some constant. The DP leading to the following theorem is based on enumeration of machine configurations.

\begin{theorem}\label{thm:KunitWeight}
If the number of scenarios~$K$ is constant, and all jobs have unit weights, then MinMaxSTC and MinAvgSTC can be solved to optimality in polynomial time on any number of machines.
\label{thm:unit-processing-times}
\end{theorem}
The proof, which can be found in Appendix~\ref{subsec:proof-conf-IP}, 
suggests an algorithm with runtime $O(mn^{2(2^K+K)})$.


\section{The MinAvg Version}
\label{sec:minsum}

By Theorem~\ref{master}, MinAvgSTC is solvable in polynomial time in the case of two scenarios and by Theorem~\ref{thm:KunitWeight} in case of a constant number of scenarios and bounded job weights. For general job weights, however, we need to design a different dynamic program (DP) that solves MinAvgSTC in polynomial time for any constant number of scenarios if there is also a constant number of machines; see Section~\ref{subsec:MinAvg-constant-machines}.

In Section~\ref{subsec:MinAvg-arbitrary-machines}, we present a conjecture that, if true, leads to a polynomial time dynamic programming algorithm for \emph{any} number of machines. We prove the conjecture for the special case of unit job weights which results in an efficient algorithm for MinAvgSTC in this case that is faster than the one given in the previous section as a function of the number of jobs, but slower as a function of the number of machines.
Moreover, it is not clear yet if the techniques carry over to the more general case of job weights bounded by a constant.



\subsection{Constant number of machines}
\label{subsec:MinAvg-constant-machines}

We first describe the case of a constant number of machines. Recall that the objective function for MinAvgSTC is defined as 
\[
\sum_{k=1}^K \sum_{i=1}^m \sum_{j\in J_i\cap S_k}w_j\cdot|\{j'\in J_i\cap S_k: j'\leq j\}|,
\]
with $J_i$ the set of jobs assigned to machine $i$.

It is clear from the objective function that the contribution of some job $j$ to the cost of a solution depends only on the assignment of that job and any jobs with higher weight, i.e., jobs $1,\dots,j-1$, to the various machines. In particular, if we want to compute the contribution of job $j$ to the cost of a solution in some schedule, it is sufficient to know the following quantities for each $i \in [m], k \in [K]$ 
\[
x_{ik}(j-1) \coloneqq |\{j'\in J_i\cap S_k: j'< j\}|, 
\]
which together form a \emph{state} of the scheduling process.
If job $j$ gets assigned to machine $i$ in that schedule, its contribution to the overall cost would then be 
\[
\sum_{k\in [K]: j\in S_k} w_j(1+x_{ik}(j-1)).
\]
In light of these observations, one can derive a dynamic program, leading to the following theorem.

\begin{theorem}
\label{thm:constant_m}
The MinAvgSTC problem with constant number of machines and constant number of scenarios can be solved in polynomial time.
\end{theorem}
\begin{proof}

We define a state in a dynamic programming decision process as a partial schedule at the moment the first $j$ jobs have been assigned and encode this by a $m\times K$ matrix $X(j)$, with $x_{ik}(j)=|\{j'\in J_i\cap S_k: j'\leq j\}|$. 

This leads to a simple DP, using the following recursion, where we use $f_j(X(j))$ to denote the minimum cost associated with the first $j$ jobs in any schedule that can be represented by~$X(j)$:
\[
f_j(X(j)) = \min_{\ell \in [m]} \left\{ f_{j-1}(X(j-1,\ell)) + \sum_{k\in [K]: j\in S_k} w_j (1 + x_{\ell k}(j-1)) \right\},
\]
where $X(j-1,\ell)$ is the matrix $X(j-1)$ from which $X(j)$ is obtained by assigning job $j$ to machine $\ell$. Equivalently, $X(j-1,\ell)$ is the matrix obtained from $X(j)$ by diminishing
all positive entries in row $\ell$ of $X(j)$ by 1. 
It follows that $X(j-1,\ell)$ has entries $x_{ik}(j-1)$ that satisfy $x_{\ell k}(j-1)=x_{\ell k}(j) -\mathbf{1}_{j \in S_k}$ for all $k$, and $x_{ik}(j-1)=x_{ik}(j)$ for all $i\neq \ell$, and all $k$.
Therefore, the computation of each state can be done in time $O(mK)$.

We initialize $f_0(\mathbf{0}) =0$ (where $\mathbf{0}$ denotes the all-zero matrix) and set $f_0(X)=\infty$ for any other possible $X$. 
Thus in each of the $n$ phases (partial job assignments) the number of possible states is bounded by $(n+1)^{m\times K}$, which, because $m$ and $K$ are constants, implies that the DP runs in polynomial time.\qed
\end{proof}

The proof implies that the problem can be solved in time $O(mKn^{mK+1})$, which is polynomial given that $m$ and $K$ are constants in this particular case.


\subsection{Any number of machines} 
\label{subsec:MinAvg-arbitrary-machines}



In this section we develop 
another efficient algorithm for MinAvgSTC on an arbitrary number of machines with a constant number of scenarios and unit job weights. 

In contrast to the dynamic program presented in the proof of Theorem~\ref{thm:KunitWeight}, the running time is linear in $n$, the number of jobs, but polynomial in $m$ with the power a function of $K$, that is, the running time is of the form $m^{h(\K)}\cdot n$ for some function~$h$.

More importantly, the technique that we use here is new and we believe it can be generalised to arbitrary job weights. For the time being it remains a fascinating open question whether MinAvgSTC can even be solved efficiently for arbitrary job weights. 

As we will explain, this is true under the assumption that the following conjecture holds, which we do believe but can only prove for the special case of unit job weights.

\begin{conjecture} 
  \label{conj:minsum}
  MinAvgSTC has an optimal solution such that for every scenario $k\in [K]$ and each $j\in [n]$, the $j$ largest jobs are assigned to the machines in such a way that the difference in number of jobs assigned to each pair of machines is bounded by a function $g(K)$ of $K$ only, or more formally 
\[ 
\max_{j\in [n], k\in [K]} \left\{  \max_{i\in [m] } |\{j'\in J_i\cap S_k: j'\leq j\}| - \min_{i \in [m]} |\{j'\in J_i\cap S_k: j'\leq j\}| \right\} \leq g(K). 
\] 
\end{conjecture}

We call the term on the left-hand side \emph{full disbalance} of the schedule. To adjust the DP to run in polynomial time for any number of machines under the conjecture, we first observe that to compute a DP recursion step, the order of the rows in each matrix $X$ representing a partial schedule is irrelevant: we only need to know how many machines have a certain number of jobs assigned to them under each scenario, not exactly which machines. 

A first step to encode partial schedules is using vectors $\ell \in \mathcal{C} \subseteq \{0,\ldots,n\}^K$, where we call~$\mathcal{C}$ the set of machine configurations. If a machine has configuration~$\ell$, it means that in the partial schedule $\ell_k$ jobs have been scheduled on this machine in scenario~$k$. 
We then simply represent a partial schedule by storing the number of machines 
that have configuration $\ell$. 
We consider the space of \emph{states} that say how many machines have a certain configuration. We can further compress this space by storing the smallest number of jobs on any machine in a given scenario separately. That is

\begin{align*}
  z_k &= \min_{i \in [m]} x_{ik},\ k \in [K]\\
\end{align*} 
The crucial observation now is that under Conjecture~\ref{conj:minsum}, there is an optimal solution such that for any partial schedule corresponding to that solution, {$0 \leq x_{ik} - z_k \leq g(K)$}. We may therefore take~$\mathcal{C}$ to be $\{0,\dots,g(K)\}^K$, i.e., the excess over $z_k$, so that $\mathcal{C}$ has constant size for constant $K$. And we define
{$$y_{\ell } = \sum_{i | x_{ik} - z_k = \ell_k,\  k \in [K]} 1,\ \ell \in \mathcal{C}.$$}\color{black} 
Since the entries of $y$ are bounded by $m$ and the entries of $z$ are bounded by $n$, the possible number of values for a pair $(y,z)$ in the encoding above is $(m+1)^{(g(K)+1)^K} (n+1)^K$, yielding a polynomial time algorithm. 
Since the DP-computation for each state in each phase (job assignment) is $O(m)$ and there are $n$ consecutive job assignments in SPT-order in the DP, this yields a polynomial time algorithm.

It is still open whether there exists an upper bound on the full disbalance of the schedule depending only on $K$. In the following, we affirm this statement for the case where all jobs have unit weights.
Note that, in this case, the $j$ largest jobs are not well-defined. We therefore prove the stronger statement that for any ordering 
of the jobs, the conjecture holds. To do so, we utilize the power of integer programming, combining the theory of Hilbert bases with techniques of an algorithmic nature. 

\begin{theorem}\label{theorem:unitconj}
    Conjecture~\ref{conj:minsum} holds for unit weights and unit processing times, where the jobs are given by an arbitrary (but fixed) order and the number $m$ of machines is part of the input.
\end{theorem}

The proof can be found in Appendix~\ref{subsec:hilbert}. 

One may wonder whether $g(K)$ can be strengthened to be polynomial in $K$. In the case of arbitrary weights, one can establish exponential lower bounds, for which we refer to Appendix~\ref{subsec:exponential-lower-bound}.


\section{Conclusion and Open Problems}
\label{sec:conc}

We hope that our results inspire interest in the intriguing field of scenario optimization problems. There are some obvious open questions that we left unanswered. For example, is it true that MinMax versions are always harder than \mbox{MinAvg} versions? For researchers interested in exact algorithms and fixed parameter tractability (FPT) results we have the following question. For the scheduling problems that we have studied so far within the scenario model, we have seen various exact polynomial time dynamic programming algorithms for a constant number of scenarios~$\K$, but always with $\K$ in the exponent of a function of the number of jobs or the number of machines. Can these results be strengthened to algorithms that are FPT in $\K$? Or are these problems W[1]-hard? Similarly, researchers interested in approximation algorithms may wonder how approximability is affected by introducing scenarios into a problem. We have given some first results here, but it clearly is a research area that has so far remained virtually unexplored. 

Another interesting variation of the MinAvg version is obtained by assigning probabilities to scenarios (i.e., a discrete distribution over scenarios) with the objective to minimize the expected total (weighted) completion time. 
The DPs underlying the results of Theorems \ref{thm:KunitWeight} and \ref{thm:constant_m} easily generalize to this version. But Theorem \ref{theorem:unitconj} does not.

We see as a main challenge to derive structural insights why multiple (constant number of) scenario versions are sometimes as easy as their single scenario versions, like the MinAvg versions of linear programming or of the min-cut problem \cite{armon2006multicriteria} or of the scheduling problem that we have studied here, and for other problems, such as MinMaxSTC, become harder or even NP-hard.

\paragraph*{Acknowledgements}
We would like to thank Bart Litjens, Sven Polak, Llu\'is Vena and Bart Sevenster for providing a counterexample for a preliminary version of Conjecture \ref{conj:minsum}, in which $g(K)$ was a linear function. Moreover, we would like to thank the anonymous referees for their valuable feedback and suggestions. Alberto Marchetti-Spaccamela was supported by the ERC Advanced Grant 788893 AMDROMA “Algorithmic and Mechanism Design Research in Online Markets”, and the MIUR PRIN project ALGADIMAR “Algorithms, Games, and Digital Markets”. Ekin Ergen and Martin Skutella were supported by the Deutsche Forschungsgemeinschaft (DFG, German Research Foundation) under Germany's Excellence Strategy --- The Berlin Mathematics Research Center MATH+ (EXC-2046/1, project ID: 390685689). Leen Stougie was supported by NWO Gravitation Programme Networks 024.002.003, and by the OPTIMAL project NWO OCENW.GROOT.2019.015.


\bibliography{sumcscenlit}   
\bibliographystyle{plain}

\section{Appendix}

\subsection{Proofs of Propositions~\ref{prop:ApproxLowerBounds} and~\ref{prop:ApproxUpperBound}}
\label{subsec:approx-proofs}

For Proposition~\ref{prop:ApproxLowerBounds} we give sketches of the proof, because of similarly proven inapproximability results in \cite{feuerstein2016minimizing}.

\begin{proof}[of Proposition~\ref{prop:ApproxLowerBounds}: Sketches]
For MinMaxSTC we can use essentially the same reduction as in the proof of Theorem~1 in \cite{feuerstein2016minimizing}. Consider a set of scenarios with $2\ell+1$ jobs each, where in each scenario the jobs can be partitioned in a perfectly balanced way. By the hardness of hypergraph balancing \cite{austrin20142+}, it is NP-hard to find a solution where none of the scenarios puts all jobs to one of the machines. In a balanced solution the cost is $\sum_{i=1}^\ell i+\sum_{i=1}^{\ell+1} i=(\ell+1)^2$. If we put all jobs on the same machine then the cost is $\sum_{i=1}^{2\ell+1} i=(2\ell+1)(\ell+1)$.

For the MinAvg version we adapt the reduction in the proof of Theorem 6 in \cite{feuerstein2016minimizing} to the total completion time objective.
  Consider a {\sc Max Cut} instance and assign to each vertex a job with weight 1 and for each edge a scenario. Then a cut gives a partition of the jobs. If the scenario is in the cut then the cost is 2, if it is not in the cut then the jobs are on the same machine and the cost is 3. Thus the objective value is 3 times the number of edges minus the size of the cut. By this observation it follows that a $(1+\alpha)$-approximation for MinAvgSTC yields a $(1-5\alpha)$-approximation for {\sc Max Cut} and our results follow from known inapproximability of {\sc Max Cut} \cite{haastad2001some,khot2007optimal}.\qed  
\end{proof}

\begin{proof}[of Proposition~\ref{prop:ApproxUpperBound}]
The result is an immediate consequence of the following well known approximation result for the classical machine scheduling problem without scenarios: If all jobs are assigned to machines independently and uniformly at random, the expected total weighted completion time of the resulting schedule is at most a factor~$\nicefrac32-\nicefrac{1}{2m}$ away from the optimum; see, e.g.,~\cite{SchulzS02,Skut-Diss98}. In our scenario scheduling model, this upper bound holds, in particular, for each single scenario which yields a randomized $(\nicefrac32-\nicefrac{1}{2m})$-approximation algorithm by linearity of expectation. Furthermore, this algorithm can be derandomized using standard techniques.\qed
\end{proof}



\subsection{Proofs of Theorems~\ref{partitionhardness} and~\ref{theorem:fptas}}
\label{appendix:MinMaxSTC}

\begin{proof}[Proof of Theorem~\ref{partitionhardness}]
We reduce Partition-3 (Partition into 3 subsets instead of 2) to MinMaxSTC. Let $\{a_1,\ldots, a_n\}$ be an instance of Partition-3, i.e., we are looking for a partition  $A_1\dot\cup A_2\dot \cup A_3= [n]$ with $\sum_{j\in A_1}a_j=\sum_{j\in A_2}a_j=\sum_{j\in A_3}a_j=\frac{1}{3}\sum_{j=1}^na_j$. 

We create $(2m+2)\cdot n$ jobs, grouping them in (disjoint) subsets $T_1,\ldots, T_n$ of size $2m+2$. For $j=1,\ldots, n$, we define the weights of the jobs in $T_j$ as follows:
\begin{enumerate}
    \item Three of the jobs have weight $Q_j$, where $(Q_j)_{j\in \{1,\ldots,n\}}$ is a decreasing sequence of sufficiently large numbers to be determined later. These jobs appear in scenarios $\{1,2\}$, $\{2,3\}$ and $\{1,3\}$ respectively. We shall refer to these jobs as \emph{white}, and depict them in figures accordingly.
    \item Three of the jobs have weight $Q_j+a_j$. These jobs also appear in scenarios $\{1,2\}$, $\{2,3\}$ and $\{1,3\}$ respectively. We refer to these jobs as \emph{black} and also depict them in figures accordingly.
    \item The remaining $2(m-2)$ jobs appear in all three scenarios and have weight $Q_j$. We refer to them as the \emph{gray} jobs.
\end{enumerate}

In particular, there are no gray jobs for $m=2$. Before determining $Q_j$, we formulate the assumptions that are needed for the reduction.
\begin{enumerate}
    \item For $j=1,\ldots, n-1$, jobs in $T_j$ are executed before those in $T_{j+1}$.
    \item Any optimal solution schedules two jobs from each phase on each machine for each scenario. 
\end{enumerate}
\begin{claim}\label{q}
For $Q_l>4mn^2a_{\max}+\sum_{j=l+1}^nm\cdot (4j-1)\cdot Q_j$, the above assumptions are satisfied.
\end{claim}
\begin{proof}
For the first condition to be fulfilled, we merely need $Q_j>Q_{j+1}+a_{j+1}$ for all $j$. 

For the second, note that we have $|S_i\cap T_j|=2m$ for each $i\in \{1,2,3\}$ and $j\in \{1,\ldots,n\}$ (two white, two black, $2(k-2)$ gray), which is a necessary condition for the second assumption. Now it suffices to prove that any solution satisfying the second property has strictly less weight than one that does not.

Any schedule that satisfies the second property places jobs from $T_j$ to have finishing times $2j-1$ and $2j$. Since the weights of the jobs in $T_j$ are upper bounded by $Q_j+a_j$, such solutions cost at most

\begin{align}\label{vgl}
    &m\cdot\sum_{j=1}^n ((2j-1)+2j)\cdot(Q_j+a_j)\leq\sum_{j=1}^nm\cdot(4j-1)\cdot(Q_j+a_{\max}).
\end{align}
Now consider a schedule that does not satisfy the second property and let $l$ be the first index where it is violated, i.e. in every scenario, the two machines are assigned two jobs from $T_j$ each for $j<l$; and there exists a scenario where one machine is assigned at least three jobs from $T_l$. For this scenario, the machine with three jobs contributes to the weight by at least 
\begin{align}\label{firstmach}
    \sum_{j=1}^{l-1} ((2j-1)+2j)\cdot Q_j+(2l-1)\cdot Q_l+2l\cdot Q_l+ (2l+1)\cdot Q_l.
\end{align}
Here, we take the sum only over elements of $T_1\cup\ldots\cup T_l$. The last summand corresponds to the third job from $T_l$ that has finishing time $2l+1$ by the minimality of $l$. Similarly, the contribution of the other machines and jobs in $T_1\cup\ldots\cup T_{l-1}$ is at least $(m-1)\cdot\sum_{j=1}^{l-1} ((2j-1)+2j)\cdot Q_j$. On the other hand, we know by $|S_i\cap T_j|=2m$ that there must be $2m-3$ other jobs in this scenario. At most $m-1$ of these jobs can be completed at time $2l-1$ because there are only $m-1$ other machines. The remaining $m-2$ jobs are completed at time at least $2l$. These two cases together contribute at least another 
\begin{equation}\label{othermach}
    ((m-1)\cdot (2l-1)+(m-2)\cdot2l) \cdot Q_l=(4ml-6l-m+1)\cdot Q_l
\end{equation} to the weight. Considering the machine loaded with at least three jobs (lower bounded by (\ref{firstmach})) as well as the other machines (lower bounded by (\ref{othermach})), the weight of the violating schedule is at least 
\begin{align}
    &\sum_{j=1}^{l-1} m\cdot(4j-1)\cdot Q_j+((4ml-6l-m+1)+6l)\cdot Q_l\\
    =&\sum_{j=1}^{l}m\cdot (4j-1)\cdot Q_j+ Q_l\\
    >&\sum_{j=1}^{l} m\cdot(4j-1)\cdot Q_j+4mn^2a_{\max}+\sum_{j=l+1}^nm\cdot(4j-1)\cdot Q_j\\
    >&\label{vgl2}\sum_{j=1}^nm\cdot(4j-1)\cdot(Q_j+a_{\max}).
\end{align}
Comparing (\ref{vgl}) and (\ref{vgl2}) yields a contradiction, finishing the proof of the claim.\qed
\end{proof}

The exact values of the $Q_j$ are not relevant for the proof. We are rather interested in their existence, which is evident given Claim~\ref{q}. One can further verify that e.g. the sequence $Q_j=(4mna_{\max})^{n-j}$ satisfies the inequality in Claim~\ref{q}, meaning that the numbers $Q_j$ can be chosen polynomially large in size of the input of Partition-3.

Since each job has unit processing time, the subsets $T_j$ correspond to a decomposition of the schedule into what we refer to as \emph{blocks}, in which elements of $T_j$ are completed in time $2j-1$ resp. $2j$ in any scenario. Therefore the contribution of the elements of $T_j$ to the sum of completion times is independent of those in $T_{j'}$ with $j'\neq j$. 

Before continuing with the reduction, we make some observations about the $2k-4$ gray jobs. By the claim above, there is no optimal solution where these jobs are assigned to the same machine as a non-gray job, for if this were the case, then there would exist a scenario where either one job or three jobs would be executed by this machine. Hence without loss of generality, we may assume that gray jobs are assigned to machines $3$ to $m$. In any optimal solution, these machines then contribute the same amount of weight, which we neglect from now on.

Let us consider the first two machines where no gray jobs are scheduled: For any assignment of the jobs, there is one scenario where two black jobs appearing in that scenario are assigned to the same machine. This is because two of the three black jobs are assigned to the same machine, and  there exists one scenario where both these jobs are included. Given that the black jobs cost more than the white ones and would ideally be completed at time $2j-1$, this creates an excess weight in this particular scenario. More precisely, let (without loss of generality) black jobs appearing in scenario $S_1$ both be assigned to the first machine. In this case, the $j$-th block contributes 
\begin{align*}
&(2j-1)\cdot\left(Q_j+a_j\right)+2j\cdot\left(Q_j+a_j\right)+{(2j-1)}\cdot Q_j+ 2j\cdot Q_j\\=&(8j-2)\cdot Q_j+(4j-1)\cdot a_j.
\end{align*}

In the scenarios $S_2$ and $S_3$ this contribution is
\begin{align*}
2\cdot(2j-1)\cdot\left(Q_j+a_j\right)+2\cdot2j\cdot Q_j=(8j-2)\cdot Q_j+(4j-2)\cdot a_j.
\end{align*}

Note that the blocks in the scenarios $S_2$ and $S_3$ are optimal and the scenario $S_1$ costs $a_j$ more than these scenarios. The three scenarios are shown Figure~\ref{block}.
\begin{figure}
    \centering
          \begin{tikzpicture}
            \node at (-0.5,1.6) {$M_1$};
            \node at (-0.5,0.5) {$M_2$};
            \node at (-0.5,-0.6) {$M_3$};
            \node at (-0.5,-2.4) {$M_k$};
            \draw [fill=black] (0,1.1) rectangle (1,2.1) node[pos=.5] {\color{white}$1,2$};
             \draw [fill=black] (1.1,1.1) rectangle (2.1,2.1)node[pos=.5] {\color{white}$1,3$};
              \draw  (2.2,1.1) rectangle (3.2,2.1)node[pos=.5] {$2,3$};            
              \draw [fill=black] (0,0) rectangle (1,1)node[pos=.5] {\color{white}$2,3$};
             \draw  (1.1,0) rectangle (2.1,1)node[pos=.5] {$1,2$};
              \draw  (2.2,0) rectangle (3.2,1)node[pos=.5] {$1,3$};
             \draw [fill=gray] (0,-0.1) rectangle (1,-1.1) node[pos=.5] {$1,2,3$};
             \draw [fill=gray] (1.1,-0.1) rectangle (2.1,-1.1)node[pos=.5] {$1,2,3$};
              \node at (1.6,-1.4) {\vdots};
          \node at (0.5,-1.4) {\vdots};
            \draw [fill=gray] (0,-1.9) rectangle (1,-2.9) node[pos=.5] {$1,2,3$};
             \draw [fill=gray] (1.1,-1.9) rectangle (2.1,-2.9)node[pos=.5] {$1,2,3$};

            \node at (4.5,1.5) {$M_1$};
            \node at (4.5,0.5) {$M_2$};
            \node at (4.5,-0.5) {$M_3$};
            \node at (4.5,-2.4) {$M_k$};

            \draw (5,0) rectangle (6,1)node[pos=.5] {$1,2$};
            \draw [fill=black] (5,1) rectangle (6,2)node[pos=.5] {\color{white}$1,2$};
             \draw  (6,0) rectangle (7,1)node[pos=.5] {$1,3$};
              \draw [fill=black] (6,1) rectangle (7,2)node[pos=.5] {\color{white}$1,3$};
             \draw [fill=gray] (5,0) rectangle (6,-1) node[pos=.5] {$1,2,3$};
             \draw [fill=gray] (6,0) rectangle (7,-1)node[pos=.5] {$1,2,3$};
                 \node at (5.5,-1.4) {\vdots};
            \node at (6.5,-1.4) {\vdots}; 
            \draw [fill=gray] (5,-1.9) rectangle (6,-2.9) node[pos=.5] {$1,2,3$};
             \draw [fill=gray] (6,-1.9) rectangle (7,-2.9)node[pos=.5] {$1,2,3$};
                
                  \draw [fill=black] (7.5,0) rectangle (8.5,1)node[pos=.5] {\color{white}$2,3$};
            \draw [fill=black] (7.5,1) rectangle (8.5,2)node[pos=.5] {\color{white}$1,2$};
             \draw (8.5,0) rectangle (9.5,1)node[pos=.5] {$1,2$};
              \draw (8.5,1) rectangle (9.5,2)node[pos=.5] {$2,3$};
             \draw [fill=gray] (7.5,0) rectangle (8.5,-1) node[pos=.5] {$1,2,3$};
             \draw [fill=gray] (8.5,0) rectangle (9.5,-1)node[pos=.5] {$1,2,3$};
            \node at (8,-1.4) {\vdots};
            \node at (9,-1.4) {\vdots};              
            \draw [fill=gray] (7.5,-1.9) rectangle (8.5,-2.9) node[pos=.5] {$1,2,3$};
             \draw [fill=gray] (8.5,-1.9) rectangle (9.5,-2.9)node[pos=.5] {$1,2,3$};

            \draw [fill=black] (10,0) rectangle (11,1)node[pos=.5] {\color{white}$2,3$};
            \draw [fill=black] (10,1) rectangle (11,2)node[pos=.5] {\color{white}$1,3$};
             \draw (11,0) rectangle (12,1)node[pos=.5] {$1,3$};
              \draw (11,1) rectangle (12,2)node[pos=.5] {$2,3$};
            \draw [fill=gray] (10,0) rectangle (11,-1) node[pos=.5] {$1,2,3$};
             \draw [fill=gray] (11,0) rectangle (12,-1)node[pos=.5] {$1,2,3$};
             \node at (10.5,-1.4) {\vdots};
            \node at (11.5,-1.4) {\vdots}; 
            \draw [fill=gray] (10,-1.9) rectangle (11,-2.9) node[pos=.5] {$1,2,3$};
             \draw [fill=gray] (11,-1.9) rectangle (12,-2.9)node[pos=.5] {$1,2,3$};

             \node at (6,-3.5) {$S_1$};   
            \node at (8.5,-3.5) {$S_2$};   
            \node at (11,-3.5) {$S_3$};   
            \end{tikzpicture}%
   \caption{Left: The assignment to machines described above. Right: The blocks for scenarios $S_1$, $S_2$ and $S_3$, respectively. The numbers denote the indices of the scenarios that the jobs belong to.}
    \label{block}
\end{figure}
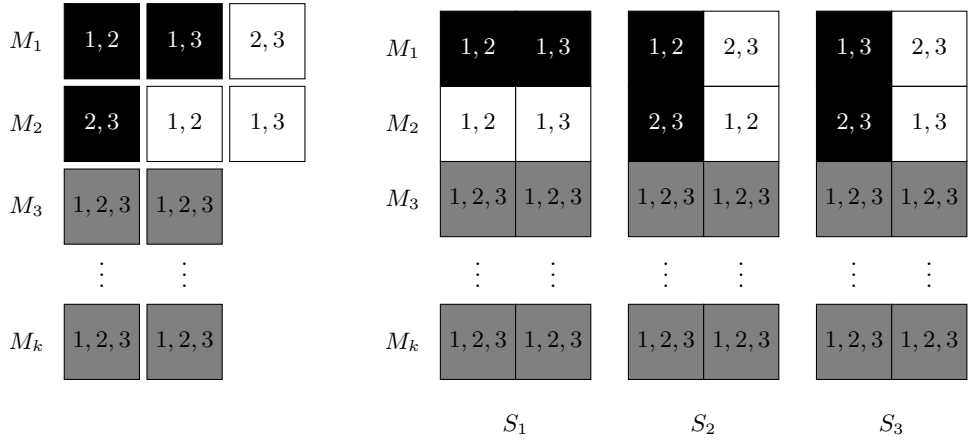

For $i\in \{1,2,3\}$ let $A_i$ be the set of blocks $j$ that cost more in the $i$-th scenario. For instance, the block in Figure~\ref{block} is an element of $A_1$. We immediately observe that $A_1\dot\cup A_2\dot\cup A_3=\{1,\ldots,n\}$. This yields for $c(J')\coloneqq \sum_{j\in J'} c_j$ given $J'\subseteq [n]$ that
\[c(S_i)=\sum_{j=1}^nc(T_j\cap S_i)=\sum_{j=1}^n((8j-2)\cdot Q_j+(4j-2)\cdot a_j)+\sum_{j\in A_i}a_j.\]
Since $c(S_1)+c(S_2)+c(S_3)=3 \sum_{j=1}^n((8j-2)\cdot Q_j+(4j-2)\cdot a_j)+\sum_{j=1}^na_j$, we have 
$$\max_{i=1,2,3}c(S_i)\geq \sum_{j=1}^n((8j-2)\cdot Q_j+(4j-2)\cdot a_j) +\frac{1}{3}\sum_{j=1}^n a_j.$$ 

This inequality is tight if and only if there is a partition of $\{1,\ldots, n\}$ into the three sets $A_1, A_2, A_3$ such that $\sum_{j\in A_1}a_j=\sum_{j\in A_2}a_j=\sum_{j\in A_3}a_j=\frac{1}{3}\sum_{j=1}^na_j$, i.e. if the set $a_1,\ldots,a_n$ is a YES-instance of Partition-3.\qed
\end{proof}

Next, we prove the FPTAS for MinMaxSTC on a constant number of machines and a constant number of scenarios. To this end, we first find a pseudopolynomial algorithm, to which the standard rounding techniques are applied.

\begin{proof}[Proof of Theorem~\ref{theorem:fptas}]
In order to develop a pseudopolynomial algorithm,
we first observe that the contribution of the $i$-th machine to the objective function value of MinMaxSTC is bounded by
\[\max_{k\in[K]} \sum_{j\in J_i\cap S_k} w_j\cdot|\{j'\in J_i\cap S_k: j'\leq j\}|\leq n^2W,\] 

where $W\coloneqq\max_{j=1}^nw_j$. 

We propose a dynamic program that, given integers $z_{ik}, y_{ik}$ with $i\in [m]$ and $k\in [K]$, decides whether there is an assignment of the jobs such that the $i$-th machine has load $y_{ik}$ and contributes exactly $z_{ik}$ to the weighted sum of completion times of the $k$-th scenario. More precisely, the dynamic program stores {assignments} \renewcommand\arraystretch{0.9}
\[\Phi^j\begin{bmatrix}y_{11}&\ldots &y_{1K} & z_{11} & \ldots & z_{1K}\\ \vdots &&\vdots &&&\vdots\\y_{m1}&\ldots&y_{mK} & z_{m1} & \ldots & z_{mK}\end{bmatrix}\colon [j]\to [m]\]
{(e.g. as vectors)} indexed by $2m K$ parameters with the property that this assignment fulfills 
\begin{equation}\label{exactval}
  \sum_{j\in J_i\cap S_k} w_j\cdot|\{j'\in J_i\cap S_k: j'\leq j\}|=z_{ik}
  \end{equation}
and 
\begin{equation}\label{exactload}
    |\{j'\in J_i\cap S_k: j'\leq{n'}\}|=y_{ik}
\end{equation}
where $J_i$ denotes the set of jobs assigned to the $i$-th machine as usual.
To allow ourselves a more compact notation, we denote these maps by $\Phi^j[Y Z]$ whenever the matrix entries are insignificant or evidently given by the matrices $Y$ and $Z$.

{The parameters $y_{ik}$ are bounded by $n$ and $z_{ik}$ can each be bounded by $n^2W+1$ because we are interested in an optimal schedule.} Note that 
{given matrices $Y\in [n]^{m\times K}$, $Z\in \{1,\ldots, n^2W\}^{m\times K}$}, an assignment $\Phi[Y Z]$
does not have to exist in general. However, as we shall see later, matrices $(Y Z)$ without corresponding assignments will not be called during the algorithm.

On the other hand, two different assignments can correspond to the same parameters $Y, Z$. This, however, will not be an issue for our purposes: If two assignments $\varphi, \varphi' \colon [n']\to[m]$ for some $n'\leq n$ both corresponded to the same parameters $y_{ik}$ and $z_{ik}$, this implies that the loads and costs of the machines are identical for the two assignments on each scenario. Therefore the dynamic program, whose main purpose is to compute an optimal assignment $\Phi^n[Y,Z]\colon [n]\to [m]$, could use either of $\varphi$ and $\varphi'$ as partial assignments, therefore we only store at most one assignment per matrix $(Y Z)$. 

For $n=1$, it is easy to decide whether a partial assignment for the given matrix $(Y Z)$ exists. This is the case if and only if 
\begin{equation}\label{dynprogiay}
     y_{ik}=
\begin{cases}
1 &i=i', 1\in J_k\\
0 &\text{else}
\end{cases}   
\end{equation}
and 
\begin{equation}\label{dynprogia}
   z_{ik}=
\begin{cases}
w_1 &i=i', 1\in J_k\\
0 &\text{else}
\end{cases} 
\end{equation}
for a fixed $i'\in [m]$. In this case, the partial assignment $\Phi^1{[Y Z]}\colon [1]\to [m]$ maps the first and only job to the $i'$-th machine.

Now let $2\leq n'\leq n$ be fixed but arbitrary. {By induction, we may assume that we have defined the assignments \[\Phi^j\begin{bmatrix}y_{11}&\ldots &y_{1K} & z_{11} & \ldots & z_{1K}\\ \vdots &&\vdots &&&\vdots\\y_{m1}&\ldots&y_{mK} & z_{m1} & \ldots & z_{mK}\end{bmatrix}\colon [j]\to [m]\] for all 
{$j\leq n'-1$} and all possible values of $(y_{ik})_{i\in[m],k\in[K]}, (z_{ik})_{i\in[m],k\in[K]}$ for which an assignment indeed exists so as to satisfy (\ref{exactload}) and (\ref{exactval}).}

Now we consider the process of assigning the $n'$-th job to a machine, which would mean an extension of an assignment $\Phi^{n'-1}[Y' Z']$ for appropriate matrices \mbox{$Y'=({y'_{ik}})_{i\in [m], k\in [K]}$} and \mbox{$Z'=({z'_{ik}})_{i\in [m], k\in [K]}$}.{ This suggests in particular that the appropriate $\Phi_{n'-1}[Y' Z']$ is defined.}

Assigning the $n'$-th job extends the map $\Phi^{n'-1}[Y' Z']$ and changes the load and contribution of the machine to which we assign $n'$ to the objective function value. That is, the resulting {assignment would be encoded by parameters $Y=(y_{ik})$ and $Z=(z_{ik})$ with entries $y_{ik}=y'_{ik}+1$ and $z_{ik}> z'_{ik}$ whenever $n'\in J_i\cap S_k$, and $y_{ik}=y'_{ik}$, $z_{ik}=z'_{ik}$ otherwise.} To be more precise, the increase of the contribution is given by \[z_{ik}-z'_{ik}=w_{{n'}}\cdot|\{j'\in J_i\cap S_k: j'\leq{n'}\}|=w_{{n'}}\cdot(y'_{ik}+1).\] 

Our dynamic program checks for $Y=(y_{ik})_{i\in[m], k\in[K]}$ and $Z=(z_{ik})_{i\in[m], k\in[K]}$ if there exists an \mbox{$m'\leq m$} such that the map 

\begin{equation}\scriptstyle
    \Phi^{n'-1}[Y' Z']\coloneqq \colon[n'-1]\to[m]
\end{equation}
where the matrices $Y'$ and $ Z'$ are defined as
\begin{equation}
    Y'=\begin{bmatrix}y_{11}&\ldots&y_{1K}\\ \vdots &&\vdots\\y_{m'1}-|\{n'\}\cap S_1| &\ldots&y_{m'1}-|\{n'\}\cap S_K|\\\vdots &&\vdots\\y_{m1}&\ldots&y_{mK}\end{bmatrix}
\end{equation} 
and 
\begin{equation}
     Z'=\begin{bmatrix}z_{11}&\ldots &z_{1K}\\ \vdots &&\vdots\\z_{m'1}-w_{n'}\cdot y_{m'1}\cdot |\{n'\}\cap S_1|&\ldots&z_{m'K}-w_{n'}\cdot y_{m'K}\cdot |\{n'\}\cap S_K|\\\vdots &&\vdots\\_{m1}&\ldots&z_{mK}\end{bmatrix}
\end{equation} 

If this is the case, we define 
$\Phi^{n'}[Y Z]\colon[n']\to[m]$ as
\[\Phi^{n'}[Y Z](j)=\begin{cases}  \Phi^{n'-1}[Z'](j) &j<n'\\ m' &j=n'\end{cases}\]

If such an $m'\leq m$ does not exist, we do not define $\Phi^{n'}[Y Z]$. This way, $\Phi^{n'}[Y Z]$ is defined by the algorithm if and only if it is the extension of another defined assignment $\Phi^{n'-1}[Y' Z']$ by the $n$-th job. By induction, the algorithm therefore defines an assignment if and only if the corresponding schedule with given $y_{ik}, z_{ik}$ is indeed realizable. The recursion base $n=1$ is given in (\ref{dynprogia}).

Note that deciding whether $m'$ exists takes linearly many calls in $m$ by linear search. Moreover, there are at most polynomially many different maps in $n$, $m$ and $W$.

Given the dynamic program, we can solve the scheduling problem as follows: Starting with $y_{\max}=z_{\max}=1$ and incrementing upwards, we enumerate all matrices $Y Z$ with $\sum_{i=1}y_{ik}=|S_k|$ for all $k\in [K]$ and $\max_{k=1}^K\sum_{i=1}^mz_{ik}=z_{\max}$, checking every time whether $\Phi^n[Y Z]$ has nontrivial domain. Once such a matrix $Z$ is found, we output the schedule $\Phi^n[Y Z]$. Since there are polynomially many matrices $[Y Z]$ in $n$ and $W$, this yields a pseudopolynomial algorithm.

We next show how to turn the pseudopolynomial algorithm into an FPTAS.

    Let {$\mathcal{I}=(w_1,\ldots, w_n, S_1,\ldots, S_K)$} be an instance of MinMaxSTC and $\varepsilon>0$. For a number $\rho$ to be determined later, we define {$w_j'\coloneqq\left\lceil \frac{w_j}{\rho}\right\rceil$.} We apply the algorithm given above to the instance $\mathcal{I'}=(w'_1,\ldots, w'_n, S_1,\ldots, S_K)$ and output the (exact) solution given by this algorithm for the instance $\mathcal{I}'$. Let {$J'_1,\ldots, J'_m$} be the partitioning of the jobs given by an optimal schedule of the instance $\mathcal{I}$. Let $\mathrm{alg}$ denote the value of the output of our proposed FPTAS that returns the optimal solution for the instance $\mathcal{I}'$, and $\mathrm{opt}$ the optimal value for the instance $\mathcal{I}$. Then we have 
    \begin{align}
       \mathrm{alg}&=\max_{k\in[K]} \sum_{i=1}^m \sum_{{j\in J'_i\cap S_k}} w_j\cdot|\{j'\in J'_i\cap S_k: j'\leq j\}| \\
        &\leq\max_{k\in[K]} \sum_{i=1}^m \sum_{{j\in J'_i\cap S_k}} w'_j\cdot \rho\cdot|\{j'\in J'_i\cap S_k: j'\leq j\}|\\
        &\leq\max_{k\in[K]} \sum_{i=1}^m \sum_{j\in J_i\cap S_k} w'_j\cdot \rho\cdot|\{j'\in J_i\cap S_k: j'\leq j\}|\\
        &< \max_{k\in[K]} \sum_{i=1}^m \sum_{j\in J_i\cap S_k} {(w_j + \rho)} \cdot|\{j'\in J_i\cap S_k: j'\leq j\}|\\
        &=\mathrm{opt}+\rho\cdot\max_{k\in[K]} \sum_{i=1}^m \sum_{j\in J_i\cap S_k} |\{j'\in J_i\cap S_k: j'\leq j\}|\leq \mathrm{opt}+\rho mn^2
    \end{align}

    On the other hand, we know that {$\mathrm{opt}\geq \max_{j=1}^nw_j\eqqcolon W$.} Therefore, 
    \[\frac{\mathrm{alg}}{\mathrm{opt}}<1+\frac{{\rho mn^2}}{\mathrm{opt}}\leq 1+\frac{{\rho mn^2}}{ W}\]

    We set $\rho\coloneqq \frac{W\varepsilon}{{mn^2}}$. Then we have $\frac{\mathrm{alg}}{\mathrm{opt}}\leq 1+\varepsilon$. Moreover, it holds that 
   {\[w'_j=\left\lceil \frac{{w_jmn^2}}{W\varepsilon}\right\rceil\leq \frac{{mn^2}}{\varepsilon} {+1},\]}
    therefore the algorithm has running time polynomial in $m$, $n$ and $\frac{1}{\varepsilon}$.\qed
\end{proof}

\subsection{Proof of Theorem~\ref{thm:KunitWeight}}
\label{subsec:proof-conf-IP}

We say that job~$j\in J$ is of \emph{type}~$T=\{k\in[K]\mid j\in S_k\}$. For~$T\subseteq[K]$, the subset of jobs of type~$T$ is denoted by~$J(T)$ and their number by~$n_T:=|J(T)|$. Notice that two jobs in~$J(T)$ are indistinguishable as they occur in exactly the same subset of scenarios and have the same (unit) weight.

Given an assignment $\phi \colon [n]\to [m]$ of the machines to the jobs, a \emph{machine configuration} is a vector~$q^{(\phi)}  =(q^{(\phi)}_T)_{T\subseteq[K]}$ with~$0\leq q^{(\phi)}_T\leq n_T$, meaning that, for every~$T\subseteq[K]$, exactly~$q^{(\phi)}_T$ jobs of type~$T$ are assigned to a machine. For the sake of simplicity, we shall suppress the assignment in the following and denote $q^{(\phi)}$ simply by $q$. 

For scenario~$k\in K$, the number of jobs scheduled on a machine with configuration~$q$ is~$n(q,k):=\sum_{T\subseteq[K]:k\in T}q_T$, and the total completion time of machine configuration~$q$ in scenario~$k$ is~$c(q,k):=\frac12 n(q,k)\bigl(n(q,k)+1\bigr)$; we define~$c_q:=\bigl(c(q,k)\bigr)_{k\in[K]}$.

The set of all possible machine configurations is denoted by~$Q$. The number of machine configurations is
\[
|Q|=\prod_{T\subseteq[K]}(n_T+1)\leq(n+1)^{2^K},
\]
and thus polynomial in the input size. A feasible solution, i.e., an assignment of jobs to machines, can be encoded by specifying, for every machine configuration~$q$, the number of machines~$y_q$ of configuration~$q$. More precisely, a feasible solution is given by a vector~$y$ satisfying
\begin{align*}
\sum_{q\in Q} y_q & = m\\
\sum_{q\in Q} q_T y_q & = n_T && \text{for all~$T\subseteq[K]$}
\end{align*}
The total completion time of solution~$y$ for scenario~$k$ is~$\sum_q c(q,k)y_q$.

\begin{proof}[of Theorem~\ref{thm:unit-processing-times}]
We show how to efficiently solve MinAvgSTC and MinMaxSTC by dynamic programming. In the following we consider tuples~$(m',\pi,d)$ with~$\pi=(\pi_T)_{T\subseteq[K]}$, $d=(d_k)_{k\in[K]}$, and
\begin{align*}
m'&\in\{0,1,\dots,m\}\\
\pi_T&\in\{0,1,\dots,n_T\} && \text{for every~$T\subseteq[K]$,}\\
d_k&\in\{0,1,\dots,\tfrac12 n(n+1)\} && \text{for every~$k\in[K]$.}
\end{align*}
Notice that the number of such tuples is
\[
(m+1)\left(\prod_{T\subseteq[K]}(n_T+1)\right)\bigl(\tfrac12 n(n+1)\bigr)^K,
\]
which is polynomial in the input size. For each tuple~$(m',\pi,d)$ define~$A(m',\pi,d)$ to be \emph{true} if there is an assignment of~$\pi_T$ jobs of type~$T$, for every~$T\subseteq[K]$, to~$m'$ machines such that the total completion time in every scenario~$k\in[K]$ is~$d_k$; otherwise,~$A(m',\pi,d)$ is \emph{false}. It thus holds that
\begin{align}
A(0,\pi,d)=
\begin{cases}
\text{true} & \text{if~$\pi_T=0$ for every~$T$ and~$d_k=0$ for every~$k$,}\\
\text{false} & \text{otherwise.}
\end{cases}
\label{eq:rec-start}
\end{align}
Moreover, for~$1\leq m'\leq m$,
\begin{align}
A(m',\pi,d)=\bigvee_{q\in Q: q\leq \pi}A(m'-1,\pi-q,d-c_q);
\label{eq:recursion}
\end{align}
here we set~$A(m'-1,\pi-q,d-c_q)$ to \emph{false} if $d-c_q\not\geq 0$. With~\eqref{eq:rec-start} and~\eqref{eq:recursion}, all values~$A(m',\pi,d)$ can be computed in polynomial time.

The value of an optimum solution to MinMaxSTC can now be determined by finding a tuple~$(m,n,d)$ with~$A(m,n,d)=\text{true}$ and~$\max_{k\in[K]}d_k$ minimum. 
Similarly, the value of an optimum solution to MinAvgSTC can be determined by finding a tuple~$(m,n,d)$ with~$A(m,n,d)=\text{true}$ and~$\sum_{k\in[K]}d_k$ minimum. 
Corresponding optimal machine assignments can be found by reverse engineering.\qed
\end{proof}

Notice that the dynamic program given in the proof above can be simplified for MinAvgSTC by only storing the total objective function value over all scenarios~$\sum_{k\in[\K]}d_k$ in the state space. Since we present a considerably faster algorithmic approach for MinAvgSTC in Section~\ref{sec:minsum}, we do not elaborate on this any further.


\subsection{Proof of Theorem~\ref{theorem:unitconj}}
\label{subsec:hilbert}

In the special case of unit weights (and unit processing times), we can prove that there is an optimal solution to every instance such that the full disbalance is bounded by a function only depending on the number $K$ of scenarios. 

To this end, we first prove a weaker statement regarding the disbalance at the end of the schedule. We shall define the \emph{final disbalance under scenario $k$} as 

\[d_k\coloneqq\max_{i\in [m]}|J_i\cap S_k|-\min_{i\in [m]}|{J_i\cap S_k}|\]

and the \emph{final disbalance} as $d\coloneqq \max_kd_k$. If the final disbalance (under a scenario $k$) refers to a solution $\varphi$ which is not clear from the context, we denote the respective final disbalance by $d(\varphi)$ ($d_k(\varphi)$).

\begin{lemma}\label{quad2}
For an instance of MinAvgSTC with unit weights and unit processing times which has a solution with {final disbalance} zero, the optimal solutions are precisely those that have {final disbalance} zero.

\end{lemma}
\begin{proof}
Let $L_k\coloneqq \frac{S_k}{m}$ be the average load per machine in the $k$-th scenario. Then 
the machines have load $L_k+a_{k,1},\ldots, L_k+a_{k,m}$ for suitable rational numbers $a_{k,1},\ldots,a_{k,m}$ with $\sum_{i=1}^ma_{k,i}=0$. The objective value for this solution then equals

\begin{equation}\label{eqn:squares}
    \sum_{k=1}^K\sum_{i=1}^n\frac{\left(L_k+a_{k,i}\right)\left(L_k+a_{k,i}+1\right)}{2}=\sum_{k=1}^K\frac{1}{2}\left(L_k^2+L_k+a_{k,i}^2\right),
\end{equation}

which is minimized if and only if $a_{k,i}=0$ for all $k\in [K]$ and $i\in [m]$. This is the case if and only if $|J_i\cap S_k|=L_k$ for all $i\in[m]$.\qed
\end{proof}

\begin{lemma}\label{lemma:dk}
    For any optimal solution of an instance of MinAvgSTC with unit weights and unit processing times, it holds that 
    \[d_k\leq\sqrt{K}\cdot 2^{K-1} \text{ for all $k\in [K]$},\]
    i.e., the final disbalance $d$ of resulting schedules with respect to any scenario is bounded by a function only dependent of the number $K$ of scenarios.
\end{lemma}
\begin{proof}
    Let an optimal solution $\varphi\colon [n]\to[m]$ be given to an instance with unit weights and processing times. Assume for the sake of contradiction that in the first scenario, the solution restricted to machines 1 and 2 (without loss of generality) admit a final disbalance strictly larger than $\sqrt{K}\cdot 2^{K-1}$, i.e.,
    \[|J_1\cap S_1|-|J_2\cap S_1|>\sqrt{K}\cdot 2^{K-1}\]
    for $J_i=\varphi^{-1}(i)$ ($i=1,2$).
    Our goal is to apply a redistribution of the jobs assigned to machines $1$ and $2$ such that the objective value is strictly decreased after the distribution, which contradicts the optimality of the initial solution. Since the assignments to the remaining machines stay unchanged, it suffices to show that the objective value of MinAvgSTC restricted to the first two machines strictly decreases after the redistribution. Hence, for the remainder of this proof, we may assume that $m=2$.

{ For the redistribution of the jobs, notice that jobs that have the same scenario profile are indistinguishable for our problem. Thus, we obtain disjoint subsets by sorting jobs according to their scenario profile. Within each subset we order the jobs in any fixed order, and assign them in this list order to the two machines, always assigning the next job to the least loaded machine within the subset, leading to an (almost) even load per machine per subset of equivalent jobs.} 

Formally, we redefine $J_1$ and $J_2$ as follows:
\begin{algorithmic}
\State $J_1, J_2\gets \emptyset$
\For {$I\subseteq[K]$, s.t. $S_I\coloneqq \bigcap_{k\in I}S_k \setminus\bigcup_{k\notin I}S_k\neq \emptyset$} 
    \State Let $S_I\coloneqq\{j^I_1,\ldots, j^I_{q_I}\}$
    \State $J_1\gets J_1\cup \{j_i^I\in S_I\colon i\text{ odd}\}$
    \State $J_2\gets J_1\cup \{j_i^I\in S_I\colon i\text{ even}\}$
\EndFor
\end{algorithmic} 

For this solution, we immediately observe that the final disbalances $d'_k$ under each scenario $k$ satisfy $d'_k\leq 2^{K-1}$ for each $k$. 
As we know that $d_1>\sqrt{K}\cdot 2^{K-1}$, we obtain
\[\sum_{k=1}^Kd_k^2\geq {d_1^2}>(\sqrt{K}\cdot 2^{K-1})^2=K\cdot 2^{2K-2}\geq \sum_{k=1}^K(d'_k)^2,\]
giving a contradiction to $\varphi$ being an optimal solution by the step~(\ref{eqn:squares}) in the proof of Lemma~\ref{quad2}.\qed
\end{proof}

Recall that the \emph{full disbalance} $f$ of a schedule is given by $f=\max_{k\in [K]}f_k$, where\[f_k\coloneqq\max_{j\in [n]} \left\{  \max_{i\in [m] } |\{j'\in J_i\cap S_k: j'\leq j\}| - \min_{i \in [m]} |\{j'\in J_l\cap S_k: j'\leq j\}| \right\}.\]

We propose an algorithm that provides a proof of the conjecture for the unit weight case and $m=2$. This algorithm will be used as a subroutine in an algorithm that provides a proof of the unit weight case (with arbitrarily many machines).

\begin{proposition}\label{prop:twoconj}
There exists an algorithm that takes an instance of MinAvgSTC on $m=2$ machines as well as an optimal solution of this instance as input, and outputs an optimal solution with full disbalance at most \[\left(2^{2^{K+1}}\cdot (K!)^2+1\right)^{2\cdot{2^K}}\cdot 2^{2^{K+1}+K+1}\cdot(K!)^2+\sqrt{K}\cdot 2^{K-1}.\]
\end{proposition}

\begin{proof}
We describe an algorithm that produces the solution as in the conjecture, starting from an arbitrary optimal solution. 


Due to some restrictions of the techniques that we apply, we aim to start with a solution $\varphi\colon [n]\to\{1,2\}$ that has zero final disbalance for each scenario, i.e. $\sum_{k=1}^Kd_k=0$, though it is evident that an optimal solution of an arbitrary instance need not satisfy this. However, by Lemma~\ref{lemma:dk}, we know that the final disbalance of any optimal solution is at most $\sqrt{K}\cdot 2^{K-1}$. To obtain a solution $\psi\colon [n+\sum_{k=1}^Kd_k]\to\{1,2\}$ with zero final disbalance, we add $d_k$ auxiliary jobs that appear only in scenario $k$ for $k=1,\ldots, K$. An optimal solution of this extended instance now has zero final disbalance at the end of {the} schedule for every scenario by Lemma~\ref{quad2}. We apply our algorithm to this extended optimal solution to obtain a solution $\psi'\colon [n+\sum_{k=1}^Kd_k]\to\{1,2\}$ and remove the auxiliary jobs at the end, obtaining an assignment $\varphi'\colon [n]\to\{1,2\}$. The following diagram demonstrates our reduction onto instances with a solution that satisfies $\sum_{k=1}^Kd_k=0$.
\usetikzlibrary{positioning}
\vspace{0.5cm}
\begin{tikzpicture}
    \node[draw, text width=3cm,minimum height=3cm,minimum width=4cm,rectangle](a) {initial optimal solution $\varphi$ of instance $\mathcal{I}$} node[above=10pt] {};
    \node[draw,text width=3cm, minimum height=3cm,minimum width=4cm,rectangle, below= 2cm of a](b){(optimal) solution $\psi$ of extended instance $\mathcal{I}'$ with zero final disbalance} ;
    \node[draw,text width=5cm,minimum height=3cm,minimum width=5cm, rectangle, right= 2cm of b](c){(optimal) solution $\psi'$ of extended instance $\mathcal{I}'$ with zero final disbalance and full disbalance bounded by a value $f(K)$ only depending on $K$} ;
    \node[draw, rectangle,  text width=5cm, minimum height=3cm,minimum width=5cm,right= 2cm of a](d){solution $\varphi'$ of instance $\mathcal{I}$ with {full disbalance} bounded by a value $g(K)$ only depending on $K$} ;

    \draw[->] (a) -- (b) node[midway, right, text width=1.7cm] {adding auxiliary jobs};
    \draw[->] (b) to node[above] {algorithm} (c);
    \draw[->] (c) -- (d) node[midway, right, text width=2cm] {removing auxiliary jobs};
     \draw[dashed,->] (a) -- (d) node[midway, left, text width=3cm] {};
\end{tikzpicture}
\vspace{0.5cm}

It is immediate that removing the auxiliary jobs increases each $f_k$ by at most $\sqrt{K}\cdot 2^{K-1}$; therefore the full disbalance $f(\varphi')$ of the solution $\varphi'$ is at most $g(K)+\sqrt{K}\cdot 2^{K-1}$ if the full disbalance of the solution $\varphi$ is $g(K)$. One also easily observes that the resulting solution is an optimal one:
\begin{claim}
    The solution $\varphi'$ as described above is optimal.
\end{claim}
\begin{proof}
    As observed in the step~(\ref{eqn:squares}) the proof of Lemma~\ref{quad2} above, it suffices to show that $\sum_{k=1}^Kd_k(\varphi')^2\leq\sum_{k=1}^Kd_k(\varphi)^2$. By our assumptions, we have $d_k(\psi')=0$ for every $k\in [K]$. Moreover, removing auxiliary jobs can create at most as much additional final disbalance for scenario $k$ as there are auxiliary jobs appearing in scenario $k$ (which is {$d_k(\varphi)$);} in total, we observe $d_k(\varphi')\leq d_k(\psi')+d_k(\varphi)=d_k(\varphi)$. Taking squares of this equation and summing over each $k$, we obtain the desired inequality.\qed
\end{proof}

Hence, in the following, we may assume that there exists a solution $\varphi$ with $\sum_{k=1}^Kd_k(\varphi)=0$. In fact, precisely the optimal solutions satisfy this by Lemma~\ref{quad2}. 

For a fixed bijection $\sigma\colon 2^{\{1,\ldots, K\}}\to \{1,\ldots, 2^K\}$, let $M\in \{0,1\}^{K\times 2^K}$ be given by entries 
\[
M_{k\sigma(S)}\coloneqq
\begin{cases}
1&k\in S,\\
0&\text{ else.}
\end{cases}
\]
This matrix helps us {to} compute how many jobs are assigned to each machine for each scenario. Let $x, y\in \mathbb{Z}_{\geq0}^{2^K}$ encode the number of jobs of each profile in machines $1$ and $2$, respectively. That is, for { $S\subseteq [K]$,} the $\sigma(S)$-th entry of $x$ resp.~$y$ denotes the number of jobs appearing precisely in scenarios $k\in S$ in machine $1$ resp. $2$. Then, the vectors $Mx, My\in \mathbb{Z}_{\geq0}^K$ have entries corresponding to the number of jobs appearing in each scenario, that is, their respective $k$-th entry denotes the number of jobs on machine $1$ resp. $2$ that appear in scenario $k\in[K]$. 

A solution $\psi\colon [n]\to\{1,2\}$ is optimal if and only if $Mx=My$. Motivated by this observation, we consider the polyhedral cone
\[
C\coloneqq\{(x,y)\in\mathbb{R}^{2\cdot2^{K}}_{\geq 0} \colon Mx-My=0\}.
\]

The set of optimal solutions is given by a subset of the lattice points $C_I\coloneqq C\cap \mathbb{Z}^{{2\cdot 2^K}}$. It is easy to observe that $C$ is a pointed rational convex cone, that is, we have
\[
C=\textrm{cone}(r^1,\ldots,r^{q(K)})
\]
where $r^1,\ldots, r^{q(K)}$ are the extreme rays of~$C$. We may assume  that $r^p\in \mathbb{Z}^{2^{K+1}}$ and 
\begin{equation}\label{factorial}
\|r^p\|_{\infty}\leq \frac{(K!)^2}{2}\qquad\text{for $p\in [q(K)]$.}
\end{equation}without loss of generality.
 As there are $2^{K+1}$ inequalities defining $C$, there can be at most $2^{2^{K+1}}$ subsystems that can define an extreme ray, therefore the number of extreme rays is bounded by~$q(K)\leq 2^{2^{K+1}}$.

An implication of $C$ being a pointed cone is that the irreducible lattice points of $C_I$ build a Hilbert basis $B$. In other words, every point in $C_I$ (and hence the encoding of every optimal solution) is an integer linear combination of the set of \emph{irreducible} vectors in $C_I$, that is, those that cannot be written as the sum of two nonzero points in $C_I$. Combinatorially, these correspond to partial solutions with final disbalance zero that do not have subsolutions with final disbalance zero. 

\begin{claim}\label{bsmall}
For any vector $b\in B$ of the Hilbert basis, we have $\|b\|_{\infty}\leq 2^{2^{K+1}}\cdot(K!)^2$.
\end{claim}
\begin{proof}
It is well-known {(cf.~\cite{schrijver_book})} that all elements of $B$ can be written as a linear combination $v=\sum_{p=1}^{q(K)}\lambda_{p}r^p$ for suitable $\lambda_p\in [0,1]$ and extreme rays $r^p$ of $C$ ($\ell\in[q(K)]$). Hence
\[\|b\|_{\infty}=\left\|\sum_{p=1}^{q(K)}\lambda_{p}r^\ell\right\|_{\infty}\leq\sum_{p=1}^{q(K)}\lambda_{p}\|r^p\|_{\infty}\leq 2^{2^{K+1}}\cdot \frac{(K!)^2}{2}\]
by~(\ref{factorial}).\qed
\end{proof}

\begin{claim}\label{basissmall}
    For the Hilbert basis $B$, it holds that $|B|\leq \left(2^{2^{K+1}}\cdot (K!)^2+1\right)^{2\cdot{2^K}}$.
\end{claim}
\begin{proof}
    By Claim~\ref{bsmall}, each entry in a basis vector admits integral values between $0$ and $2^{2^{K+1}}\cdot(K!)^2$. There are hence $2^{2^{K+1}}\cdot(K!)^2+1$ choices for each of the $2\cdot2^K$ entries.\qed
\end{proof}

Now consider an optimal schedule (i.e. one with final disbalance zero) and the corresponding vector~$v\in C_I$. Then, by definition of a Hilbert basis, we find a decomposition 
\begin{equation}\label{lambda}
   v=\sum_{b\in B}\lambda_bb  
\end{equation}
with $\lambda_b\in \mathbb{N}$ and $\|b\|_{\infty}\leq 2^{2^{K+1}}\cdot\frac{(K!)^2}{2}$. This means that we can decompose jobs in our schedule into $\sum_{b\in B}\lambda_b$ classes each of which has final disbalance zero. These classes correspond to the basis elements $b\in B$, there are $\lambda_b$ copies of such classes for each $b\in B$.

The core idea of the algorithm is to exploit the fact that $|B|$ is bounded by a function of $K$. Since $\lambda_b$ are not necessarily bounded by a function of $K$, the most important detail is to bound the sum of full disbalances of classes that contribute to the same $b\in B$.

To describe the algorithm, we first consider a fixed $b\in B$. Let $\tau\colon [n]\to 2^{S}, j\mapsto \{s\in S\colon j\in s\}$ be the transversal to the scenario sets. Then each job $j$ contributes to the vector $v$ by a unit vector; more precisely by $e_{\sigma(\tau(j))}$ if it is assigned to the first machine and by $e_{2^K+\sigma(\tau(j))}$ if it is assigned to the second.
Our algorithm takes the jobs that contribute to the summand $\lambda_bb$ of the sum~(\ref{lambda}) according to the order $\sigma$ and assigns each job to the class with the smallest possible index in which $\tau(j)$ ``fits''. 

More precisely, let $B=\{b_1,\ldots, b_{|B|}\}$ and let $\psi\colon [n]\to \{1,2\}$ be an asignment with zero final disbalance which corresponds to a vector $v=\sum_{p=1}^B\lambda_{b_p}b_p  \in C_I$ and let $b_p^1,\ldots, b_p^{\lambda_{b_p}}$ denote copies of $b_p$.  We describe the new assignment \mbox{$\psi'\colon [n]\to\{1,2\}$} as follows: 
\newcommand{\algorithmicbreak}{\textbf{continue}}
\newcommand{\Continue}{\State \algorithmicbreak}

\begin{algorithmic}

\For {$j=1$ to $n$}
\For{$p=1$ to $|B|$}
    \For{$\ell=1$ to $\lambda_{b_p}$}
    \If {$b_p^\ell-e_{\sigma(\tau(j))}\geq 0$}
        \State $b_p^\ell\gets b_p^\ell-e_{\sigma(\tau(j))}$
        \State $\psi(j)\gets 1$
        \Continue
    \ElsIf{$b_p^\ell-e_{2^K+\sigma(\tau(j))}\geq 0$}
        \State $b_p^\ell\gets b_p^\ell-e_{2^K+\sigma(\tau(j))}$
        \State $\psi(j)\gets 2$
        \Continue
    \EndIf
    
    \EndFor
\EndFor
\EndFor
\end{algorithmic}  

\begin{claim}
At any point of the algorithm, we have $b_p^\ell\leq b_p^{\ell'}$ for $\ell<\ell'$.
\end{claim}
\begin{proof}
We prove the claim via induction over the jobs $j$. At the beginning of the algorithm, $b_p^\ell$ and $b_p^{\ell'}$ both equal $b_p$, therefore the claim holds trivially. 

Let $j$ be an arbitrary job that is assigned to a machine at the $p$-th iteration of the second inner loop and the $\ell'$-th iteration of the innermost loop. By induction, we may assume that $b_p^\ell\leq b_i^{\ell'}$ right before the $j$-th iteration of the outermost loop. This is the only case where $b_p^{\ell'}$ is decreased, and since the values of $b_p^\ell$ and $b_p^{\ell'}$ are monotonically decreasing throughout the algorithm, such jobs $j$ are the only reasons that the claim could become violated for the first time. Since $\ell<\ell'$, the job $j$ has not been assigned at the $\ell$-th iteration, meaning that $(b_p(\ell))_{\sigma(\tau(j))}=(b_p(\ell))_{2^K+\sigma(\tau(j))}=0$ at the $j$-th iteration of the outermost loop. On the other hand, we have $(b_p(\ell))_{\sigma(\tau(j))},(b_p(\ell))_{2^K+\sigma(\tau(j))}\geq0$ at the end of the $j$-th iteration, for else the assignment for $j$ would have happened later. Since all other entries of $b_p^\ell$ and $b_p^{\ell'}$ stay constant during this iteration, the claim follows.\qed
\end{proof}

The monotonicity of the $b_p^\ell$ in the parameter $\ell$ implies that at any point of the algorithm, there is at most one $\ell$ such that $(b_p^\ell)_k-(b_p^\ell)_{k+2^K}\neq0$ holds for some~$k\in [K]$. This particular $\ell\eqqcolon \ell(j)$ is the index of the only copy whose corresponding jobs contribute to the current final disbalance. To be more precise,
we have for every $j\in [n]$ and $k\in[K]$:
\begin{align*}
    &\max_{i\in \{1,2\} } |\{j'\in J_i\cap S_k\cap B_p: j'\leq j\}| - \min_{i \in \{1,2\}} |\{j'\in J_i\cap S_k\cap B_p: j'\leq j\}|\\ \leq & \max_{i\in \{1,2\} } |\{j'\in J_i\cap S_k\cap B_{p\ell(j)}: j'\leq j\}| - \min_{i \in \{1,2\}} |\{j'\in J_i\cap S_k\cap B_{p\ell(j)}: j'\leq j\}|
\end{align*}
 where $B_p$ resp. $B_{p\ell(j)}$ denotes the set of jobs that are assigned during the middle iteration $p$ resp. the two innermost iterations $(p,\ell(j))$. The size of the latter is bounded by $\|b_p\|_1$ by the construction of the algorithm.

Therefore, we have
\begin{align}
& 
\max_{j\in [n], k\in [K]} \left\{  \max_{i\in \{1,2\} } |\{j'\in J_i\cap S_k: j'\leq j\}| - \min_{i \in \{1,2\}} |\{j'\in J_i\cap S_k: j'\leq j\}| \right\} 
\\
\leq \sum_{p=1}^{|B|}&\max_{j\in [n], k\in [K]} \left\{  \max_{i\in \{1,2\} } |\{j'\in J_i\cap S_k\cap B_p: j'\leq j\}| - \min_{i \in \{1,2\}} |\{j'\in J_i\cap S_k\cap B_p: j'\leq j\}| \right\}\\
\leq\sum_{p=1}^{|B|} &\max_{j\in [n],k\in [K]} \left\{  \max_{i\in \{1,2\} } |\{j'\in J_i\cap S_k\cap B_{p\ell(j)}: j'\leq j\}| - \min_{i \in \{1,2\}} |\{j'\in J_i\cap S_k\cap B_{p\ell(j)}: j'\leq j\}| \right\} \\
\leq \sum_{p=1}^{|B|}&\|b_p\|_1\leq \sum_{p=1}^{|B|}2\cdot 2^K\cdot\|b_p\|_{\infty}\leq\left(2^{2^{K+1}}\cdot (K!)^2+1\right)^{2\cdot{2^K}}\cdot 2\cdot 2^K\cdot 2^{2^{K+1}}\cdot(K!)^2
\end{align}
The last inequality holds by Claims~\ref{bsmall} and~\ref{basissmall}.      
 This inequality together with the reduction at the beginning finishes the proof.\qed
\end{proof}

\begin{proof}[Proof of Theorem~\ref{theorem:unitconj}]
    We use the algorithm from Proposition~\ref{prop:twoconj} as a subroutine to apply on two machines, after which we obtain full disbalance bounded by a function $f(K)$ when restricted to these two machines (cf.\ proof of Proposition~\ref{prop:twoconj}). We call this subroutine \emph{equalizing} and denote by $\mathbf{equalize}(i_1,i_2)$ when applied to machines $i_1$ and $i_2$. We equalize two machines at a time and prove that for a suitable function value $g(K)$ depending on the number $K$ of scenarios as well as a suitable choice of pairs of machines to equalize, the procedure eventually terminates with the desired outcome. 

    Let $L_k\coloneqq \frac{S_k}{m}$ be the average load of a machine for scenario $k\in [K]$. Then the procedure can be described as follows: 

    \begin{algorithmic}

    \While{there exists $i\in [m]$, $k\in [K]$ with $||J_i\cap S_k|-L_k|>2K\cdot f(K)$}
        \If{$|J_i\cap S_k|>L_k$}
         \State \textbf{equalize}($i,\mathrm{argmin}_{i'\in[m]}|J_{i'}\cap S_k|$)
         \Else
        \State \textbf{equalize}($i,\mathrm{argmax}_{i'\in[m]}|J_{i'}\cap S_k|$)
        \EndIf
    \EndWhile
    \end{algorithmic}  

    By construction, the algorithm outputs a solution with full disbalance at most $4K\cdot f(K)$ if it terminates, because then we have for any two machines $i_1,i_2\in [m]$ and for any scenario $k\in [K]$:
    \begin{align*}
   ||J_{i_1}\cap S_k|-|J_{i_2}\cap S_k||&= \left|(|J_{i_1}\cap S_k|-L_k)-(|J_{i_2}\cap S_k|-L_k)\right|    \\
   &\leq ||J_{i_1}\cap S_k|-L_k|+||J_{i_2}\cap S_k|-L_k|\leq 2\cdot 2K\cdot f(K)
   \end{align*}

   Therefore, it suffices to show that the procedure terminates. 
   \begin{claim}
         The sum $\mathbf{L}\coloneqq\sum_{i=1}^m\sum_{k=1}^K ||J_i\cap S_k|-L_k|$ strictly decreases at each iteration of the procedure.  
   \end{claim}
    \begin{proof}
        After the subroutine $\textbf{equalize}(i_1,i_2)$ is applied, all summands of $\mathbf{L}$ indexed by $i\notin\{i_1, i_2\}$ stay constant. Let $\mathbf{k}$ be the scenario because of which the subroutine was applied (i.e.~one has $||J_{i_1}\cap S_k|-L_k|>2K\cdot f(K)$ or $||J_{i_2}\cap S_k|-L_k|>2K\cdot f(K)$). Then the summands $||J_{i_1}\cap S_{\mathbf{k}}|-L_{\mathbf{k}}|$ resp. $||J_{i_2}\cap S_{\mathbf{k}}|-L_{\mathbf{k}}|$ were each decreased by at least $\frac{1}{2}||J_{i_1}\cap S_{\mathbf{k}}|-L_{\mathbf{k}}|-f(K)$ resp. $\frac{1}{2}||J_{i_1}\cap S_{\mathbf{k}}|-L_{\mathbf{k}}|-f(K)$. For $k'\neq\mathbf{k}$, the increase each summand  $||J_{i_1}\cap S_{k'}|-L_{k'}|$ resp. $||J_{i_2}\cap S_{k'}|-L_{k'}|$ is at most $f(K)$. Therefore, the total decrease is at least 
        \begin{align*}
            &\frac{1}{2}||J_{i_1}\cap S_{\mathbf{k}}|-L_{\mathbf{k}}|-f(K)+\frac{1}{2}||J_{i_1}\cap S_{\mathbf{k}}|-L_{\mathbf{k}}|-f(K)-2(K-1)\cdot f(K)\\
            >&\frac{1}{2}\cdot 4K\cdot f(K)- 2K\cdot f(K)=0.
        \end{align*}
    \end{proof}
    The claim implies that the procedure must terminate after finitely many steps as we have $\mathbf{L}\geq 0$ throughout the procedure.\qed
\end{proof}

\subsection{Exponential Lower Bound on Disbalance}
\label{subsec:exponential-lower-bound}

Let us turn to the lower bound construction. First, we consider the following related question. Suppose we are given a $0-1$ matrix $A\in\{0,1\}^{n\times K}$ such that every column in $A$ sums to $c$. We now want to know whether there exists a proper $n'\times K$ submatrix $A'$ of $A$ such that each column in $A'$ sums to $c'<c$. If this is not the case, we call matrix $A$ unsplittable.

For every unsplittable matrix $A$ such that every column sums to $c$, we can create an instance of MinAvgSTC such that the disbalance of the schedule is $c$. To see this, note that we may create an instance on 2 machines, consisting of one job for every row in $A$ with weight $1$, and a scenario for each column in $A$, where such a job $j$ belongs to scenario $k$ if and only if there is a $1$ in the corresponding entry. Furthermore, we create another set of $c$ jobs appearing in every scenario, with weight $1 - \epsilon$. For $\epsilon$ small enough in the optimal solution the jobs corresponding to the rows of $A$ should be assigned to the same machine, say machine 1, and the $c$ additional jobs to machine 2. However, the DP will schedule the rows of $A$ first and therefore after having assigend all jobs corresponding to the rows of $A$ it has to create a disbalance of $c$, since machine 2 will have been assigned no jobs yet.



We can now construct a family of instances showing that the disbalance of the schedule is at least exponential in $K$. Let $q$ be some natural number and let $I$, $J$ be the $q \times q$ identity respectively the all ones matrix, and let $H$ be $J - I$. We define the following $q^2 
\times 2q$ matrix.

\setlength\arraycolsep{2pt}
\renewcommand{\arraystretch}{1}
\[ A_q^2 =	\begin{pmatrix}
		I & J \\
		J & H \\
		\vdots & \vdots \\
		J & H
		\end{pmatrix}, \]

where the submatrix $(J ~~ H)$ is repeated $q-1$ times. Now the columns sum to $q^2 - q + 1$, but no proper submatrix has columns summing to the same number. To see this, note that if we take a row from the top $(I ~~ J)$, we must take all of them, just to balance out the first $q$ columns. But then the first $q$ columns contain in total $q^2 - q$ fewer ones than the last $q$ columns. As every row in the bottom has exactly one zero, it takes all of the bottom rows to make up the difference.

We can use this insight to work up the example. Let $J_{m,n}$ be the $m \times n$ all ones matrix. Let $B_q^2 = J_{q^2,2q} - A_q^2$. Then $B_q^2$ has no proper submatrix with uniform column sums either, but has every column summing to $q-1$. Hence, the following matrix is again unsplittable:

\[ A_q^3 =	\setlength\arraycolsep{2pt}\begin{pmatrix}
		B_q^2 & J_{q^2,q} \\
		J_{q,2q} & H \\
		\vdots & \vdots \\
		J_{q,2q} & H
		\end{pmatrix}, \]

where the submatrix $(J_{q,2q} ~~ H)$ is repeated $q^2 - 2q + 2$ times. Note that the number of ones in $B_q^2$ is equal to $2q^2 -2q$, whereas the number of ones in $J_{q^2,q}$ equals $q^3$. Since the difference is a multiple of $q$, we can add an integer number of copies of $(J_{q,2q} ~~ H)$ to obtain the unsplittable matrix $A_q^3$. This matrix has columns summing to $q^3 - 2q^2 +3q -1$. Repeating this construction will result in matrix $A_q^t$, for $t\geq 4$, for which the following theorem holds. 

\begin{theorem}
Using the construction above, we obtain the unsplittable matrix $A_q^t$, for $t\geq 2$, satisfying the following properties:
\begin{enumerate}
    \item The number of columns equals $tq$,
    \item The number of rows is a polynomial in $q$ of degree $t$, where the leading coefficient equals 1,
    \item The columns sum to a polynomial in $q$ of degree $t$, where the leading coefficient equals 1.
\end{enumerate}
\end{theorem}
\begin{proof}
We will proof the theorem using induction on $t$. We have already shown that matrix $A_q^2$ satisfies the required properties. Now, we assume the theorem holds for matrix $A_q^{t-1}$. Let the number of rows of $A_q^{t-1}$ be equal to $q^{t-1}+p_1(q)$, and let the column sum be equal to $q^{t-1}+p_2(q)$, where $p_1(q)$ and $p_2(q)$ are polynomials in $q$ of degree $t-2$.

Let $B_q^{t-1}=J_{q^{t-1}+p_1(q),(t-1)q} - A_q^{t-1}$. Since $A_q^{t-1}$ is unsplittable, so is $B_q^{t-1}$. Each column of $B_q^{t-1}$ sums to
\[ q^{t-1} + p_1(q) - (q^{t-1} + p_2(q)) = p_1(q)-p_2(q), \]
which is a polynomial of degree $t-2$. In our construction, we would like to add $J_{q^{t-1}+p_1(q),q}$ to the right of $B_q^{t-1}$, and add a number of copies of $(J_{q,(t-1)q} ~~ H)$ to balance the column sums. The resulting matrix is called $A_q^t$.

Note that the number of ones in $B_q^{t-1}$ equals $(t-1)q(p_1(q)-p_2(q))$, whereas the number of ones in $J_{q^{t-1}+p_1(q),q}$ equals $q^t+qp_1(q)$. Since the difference between these numbers is a multiple of $q$, we can actually add an integer number of copies of $(J_{q,(t-1)q} ~~ H)$ to balance the columns. Now, matrix $A_q^t$ has $tq$ columns, and the number of rows is equal to
\[ q^{t-1} + p_1(q) + q^t + qp_1(q) - (t-1)q(p_1(q)-p_2(q)), \]
which is a polynomial in $q$ of degree $t$, where the leading coefficient equals 1. Similarly, each column of $A_q^t$ sums to
\[ p_1(q)-p_2(q) + q^t + qp_1(q) - (t-1)q(p_1(q)-p_2(q)), \]
which is a polynomial in $q$ of degree $t$, where leading coefficient equals 1. Hence, $A_q^t$ also satisfies the properties stated in the theorem.\qed
\end{proof}

Since $A_q^t$ has $tq$ columns summing to $\Omega(q^t)$, we obtain an $\Omega(q^{K/q})$ lower bound.

\subsection{Consequences of our Results for Regret Scheduling}
\label{subsec:regret}

In the proof of Theorem~\ref{partitionhardness}, we applied our reduction by building very specific instances of MinMaxSTC. This allows us to conclude that MinMaxSTC on a restricted class of instances is also NP-hard.

\begin{corollary}
Let an instance of MinMaxSTC with unit processing times be called \emph{scenario-symmetric} if for any permutation of the scenarios, the sets of unnumbered jobs are identical as sets of tuples of weight and scenarios, i.e. we have for any permutation $\pi\colon[K]\to[K]$
  \[\{(w_i, \{k\in[K]\colon i\in S_k\})\}_{i\in[n]}=\{(w_i, \{\pi(k)\in[K]\colon i\in S_k\})\}_{i\in[n]}\]
  MinMaxSTC restricted to scenario-symmetric instances with unit processing times is NP-hard.
\end{corollary}
\begin{proof}
    In the proof of Theorem~\ref{partitionhardness}, the instance in the reduction is scenario-symmetric because gray jobs appear in all three scenarios, while the three black and three white jobs in each block appear in scenarios $\{1,2\},\{1,3\}, \{2,3\}$ respectively, which is invariant under permutations of the three scenarios. Hence any polynomial-time algorithm for MinMaxSTC on scenario-symmetric instances implies a polynomial-time algorithm for Partition-3.\qed
\end{proof}

We can thus relate our model to robust min-max regret scheduling with arbitrary weights and unit processing times~\cite{shabtay2022}. Indeed, the latter model seeks to minimize

\[\max_{k=1}^K \left( G(\sigma,k)-\min_{\tau\colon [n]\to [m]}G(\tau,k)\right),
\]

while our model seeks to minimize $\max_{k=1}^K \left( G(\sigma,k)\right)$, where $G(\tau,k)$ denotes the sum of completion times of the assignment $\tau$ in scenario $k$. In general, these two objective functions do not have a one-to-one correspondence. However, under the assumption that the scenarios are scenario-symmetrical, optimal solutions of the two models coincide. 

\begin{corollary}
Scheduling with regret on $m\geq 2$ machines, $K=3$ scenarios and unit processing times is NP-hard.
\end{corollary}

Moreover, the dynamic program used in the proof of Theorem~\ref{theorem:fptas} can also be used to solve min-max regret scheduling by first computing $\min_{\tau\colon [n]\to [m]}G(\tau,k)$ (where $k$ is trivially the only scenario) and then finding assignments to match every possible value of the objective value. It is then immediate that Theorem~\ref{theorem:fptas} can be applied as well, proving that robust min-max regret scheduling admits an FPTAS.

One can also consider scheduling with regret with respect to the sum over scenarios, that is, the model minimizing 

\[\sum_{k=1}^K \left( G(\sigma,k)-\min_{\tau\colon [n]\to [m]}G(\tau,k)\right)=\sum_{k=1}^K  G(\sigma,k)-\sum_{k=1}^K\min_{\tau\colon [n]\to [m]}G(\tau,k)\]

Once again, the latter term does not depend on the solution. Therefore, it immediately follows that the optima of this problem coincide with those of MinAvgSTC.

\end{document}